\documentclass[10pt,twocolumn,twoside]{IEEEtran}
\usepackage{amsmath}
\usepackage{color}
\usepackage{cases}
\usepackage{amsfonts,amsmath,amssymb}
\usepackage{mathrsfs}
\usepackage{cite}
\usepackage{graphicx}
\usepackage{url}
\usepackage{enumerate}
\usepackage{subfigure}
\usepackage{hyperref}
\usepackage{algorithm}
\usepackage{algorithmic}
\usepackage{amsmath,bm}
\usepackage{caption}
\usepackage{float}
\usepackage{stfloats}
\usepackage{multirow}
\usepackage{url}
\usepackage{breakurl}
\usepackage{txfonts}
\usepackage{booktabs}
\usepackage{threeparttable}
\usepackage{multirow}
\usepackage{fancyhdr}
\usepackage{lettrine}

\newtheorem{theorem}{Theorem}
\newtheorem{proposition}{Proposition}
\newtheorem{lemma}{Lemma}
\newtheorem{corollary}{Corollary}
\newtheorem{example}{Example}

\DeclareMathOperator*{\argmin}{arg\,min}

\begin{document}

\title{Robust Multidimentional Chinese Remainder Theorem for Integer Vector Reconstruction}

\author{\mbox{Li Xiao, Haiye Huo, and Xiang-Gen Xia,~\IEEEmembership{Fellow,~IEEE}}

\thanks{This work was supported in part by the National Natural Science Foundation of China under Grants 62202442, 12261059, and 11801256, in part by the Anhui Provincial Natural Science Foundation under Grant 2208085QF188, in part by the Jiangxi Provincial Natural Science Foundation under Grant 20224BAB211001, and in part by the US National Science Foundation under Grant CCF-2246917.}
\thanks{L. Xiao is with the Department of Electronic Engineering and Information Science, University of Science and Technology of China, Hefei 230052, China, and also with the Institute of Artificial Intelligence, Hefei Comprehensive National Science Center, Hefei 230088, China (e-mail:
xiaoli11@ustc.edu.cn).}
\thanks{H. Huo is with School of Mathematics and Computer Sciences, Nanchang University, Nanchang 330031, China (e-mail: hyhuo@ncu.edu.cn).}
\thanks{X.-G. Xia is with the Department of Electrical and Computer Engineering, University of Delaware, Newark, DE 19716, USA (e-mail: xxia@ee.udel.edu).}
}

\maketitle

\begin{abstract}
The problem of robustly reconstructing an integer vector from its erroneous remainders appears in many applications in the field of multidimensional (MD) signal processing. To address this problem, a robust MD Chinese remainder theorem (CRT) was recently proposed for a special class of moduli, where the remaining integer matrices left-divided by a greatest common left divisor (gcld) of all the moduli are pairwise commutative and coprime. The strict constraint on the moduli limits the usefulness of the robust MD-CRT in practice. In this paper, we investigate the robust MD-CRT for a general set of moduli. We first introd-\\uce a necessary and sufficient condition on the difference between paired remainder errors, followed by a simple sufficient condition on the remainder error bound, for the robust MD-CRT for gener-\\al moduli, where the conditions are associated with (the minimum distances of) these lattices generated by gcld's of paired moduli, and a closed-form reconstruction algorithm is presented. We then generalize the above results of the robust MD-CRT from integer vectors/matrices to real ones. Finally, we validate the robust MD-CRT for general moduli by employing numerical simulations, and apply it to MD sinusoidal frequency estimation based on multiple sub-Nyquist samplers.

\end{abstract}


\section{Introduction}\label{sec1}
The Chinese remainder theorem (CRT) \cite{crt_integer} is known to offer a solution to a system of linear congruence equations, namely, reconstructing a larger nonnegative integer from its remainders modulo several smaller positive integers (called moduli). It has a broad range of applications in many areas, such as computer arithmetic, digital signal processing, and cryptography \cite{crt_integer,CRT1,CRT2}. Nevertheless, the CRT reconstruction is extremely susceptible to errors in the remainders, in the sense that a very small error in any remainder might yield a large reconstruction error in the large integer of interest. This may cause failures in applications of the CRT, considering that the detected remainders are often erroneous due to environmental noise contamination. As such, during the past decades, the problem of robust reconstructions from the erroneous remainders has been continuously investigated, where ``robustness'' means that the reconstruction error is upper bounded by the remainder error bound \cite{xia1-2007,xixiaixiao,wjwang2010,binyang2014,xiaoxia1,wenjiemaximum,xiaoxia7}. More specifically, for addressing this robust remaindering problem, a robust CRT has been introduced, of which the basic idea is to accurately determine all the quotients (called folding numbers) of the large integer divided by the moduli. A thorough review of the robust CRT and its various generalizations is presented in \cite{xiaoli}. To distinguish from the robust multidimensional (MD) CRT for integer vector reconstruction studied in this paper, we refer to the robust CRT for integer reconstruction as the robust 1-D CRT, which has been found to have potential applications to sinusoidal frequency estimation with sub-Nyquist samplings and phase unwrapping for radar interferometry \cite{mruegg2007,gli1-2007,yimin2008,xwli1-2011,grslet2013,aa2015}, error control neural coding \cite{fiete}, signal recovery using multi-channel modulo samplers \cite{ganganlulu}, and wireless sensor networks with fault tolerance \cite{fiete2,chessa,yishengsu}.

Considering that signals found in modern applications often have a multidimensional structure, e.g., multiple input multiple output (MIMO) radar and MIMO communication systems, we recently studied exact and robust reconstructions of an integer vector from its (erroneous) remainders modulo several moduli in \cite{md-crt}, where the moduli are nonsingular integer matrices and the remainders are integer vectors. Concretely, we first derived the MD-CRT for a general set of moduli, via which an integer vector can be accurately reconstructed from the remainders, if this integer vector is within the fundamental parallelepiped of the lattice that is generated by a least common right multiple of all the moduli. We then introduced the robust MD-CRT for a special class of moduli, where these remaining integer matrices left-divided by a greatest common left divisor of all the moduli are pairwise commutative and coprime. In this special case, the robust MD-CRT basically states that an integer vector within a certain reconstruction range can be robustly reconstructed from its erroneous remainders and the moduli, if the remainder error bound is smaller than a quarter of the minimum distance of the lattice that is generated by a greatest common left divisor of all the moduli. One can clearly see that there is the commutativity and coprimeness constraint on the matrix moduli for the robust MD-CRT in \cite{md-crt}, which might be too strong and therefore may limit the applications of the robust MD-CRT in practice.

In this paper, we propose the robust MD-CRT for a general set of moduli on which the undesirable matrix commutativity and coprimeness constraint we imposed in \cite{md-crt} is no longer required. Instead of accurately determining the folding vectors $\left\{\textbf{n}_i\right\}_{i=1}^L$ (namely, the quotients of an integer vector of interest $\textbf{m}$ left-divided by moduli $\left\{\textbf{M}_i\right\}_{i=1}^L$) in \cite{md-crt}, we attempt to accurately determine $\left\{\textbf{M}_i\textbf{n}_i\right\}_{i=1}^L$, and thereby obtain a robust reconstruction $\tilde{\textbf{m}}$ of $\textbf{m}$  by averaging the reconstructions calculated from the determined folding vectors, i.e., $\tilde{\textbf{m}}=\frac{1}{L}\sum_{i=1}^{L}\left(\textbf{M}_{i}\textbf{n}_{i}+\tilde{\textbf{r}}_{i}\right)$, in this paper, where $\left\{\tilde{\textbf{r}}_{i}\right\}_{i=1}^L$ denote the erroneous remainders. Of note, this strategy actually facilitates the robust MD-CRT for a general set of moduli by avoiding the difficulties brought about by the non-commutativity of matrix multiplication. More precisely, we first present a necessary and sufficient condition on the difference between paired remainder errors, as well as a simple sufficient condition on the remainder error bound, for the robust MD-CRT for general moduli, where the conditions are related with (the minimum distances of) the lattices that are generated by greatest common left divisors of paired moduli. At the same time, a closed-form reconstruction algorithm for the derived robust MD-CRT is proposed as well. In addition, we generalize the above results of the robust MD-CRT from integer vector/matrix cases to real-valued vector/matrix cases. We finally validate the robust MD-CRT for general moduli by conducting some numerical simulations, and apply it to frequ-\\ency estimation for a complex MD sinusoidal signal undersampled with multiple sub-Nyquist samplers. It demonstrates that the use of the robust MD-CRT with $L$ properly chosen moduli $\left\{\textbf{M}_i\right\}_{i=1}^L$ (whose inverses are referred to as sub-Nyquist sampling matrices with sampling densities $\left\{|\text{det}(\textbf{M}_i)|\right\}_{i=1}^L$) can result in significant sampling density reduction over the Nyquist sam-\\pling density for MD sinusoidal frequency estimation.

The rest of this paper is organized as follows. We introduce the preliminary knowledge associated with integer vectors and integer matrices in Section \ref{sec2}, as well as our previously derived (robust) MD-CRT in Section \ref{sec3} where the robust MD-CRT is limited to a special class of moduli. In Section \ref{sec4}, we propose the robust MD-CRT for a general set of moduli, together with its closed-form reconstruction algorithm. We further generalize the robust MD-CRT from integer vectors/matrices to real ones in Section \ref{sec5}. We present simulation results of the robust MD-CRT and its application to MD sinusoidal frequency estimation with multiple sub-Nyquist samplers in noise in Section \ref{sec6}. We conclude this paper in Section \ref{sec7}.

\textit{Notations}: We utilize capital and lowercase boldfaced letters to denote matrices and vectors, respectively. Let $A(i,j)$ be the $(i,j)$-th element of a matrix $\textbf{A}$, and $a(i)$ be the $i$-th element of a vector $\textbf{a}$. Let $\textbf{A}^T$, $\textbf{A}^{-1}$, $\textbf{A}^{-T}$, and $\text{det}(\textbf{A})$ denote the transpose, inverse, inverse transpose,
and determinant of $\textbf{A}$, respectively. We represent by $\text{diag}(a_1,a_2,\cdots,a_D)$ the diagonal matrix with a scalar $a_i$ being the $i$-th diagonal element. Let $\mathbb{R}$ and $\mathbb{Z}$ denote the sets of reals and integers, respectively. For a $D$-dimensional real vector $\textbf{a}\in\mathbb{R}^D$, $\textbf{a}\in[c,d)^D$ says that every element of $\textbf{a}$ is within the range of $[c,d)$ and $c,d\in\mathbb{R}$. Let $\textbf{I}$ and $\textbf{0}$ respectively be the identity matrix and the all-zero vector/matrix (their sizes are determined from the context). The symbol $\lfloor\cdot\rfloor$ denotes the floor operation, and it is implemented element-wisely if acting on one vector. We let $\text{adj}(\textbf{M})$ stand for the adjugate of a square matrix $\textbf{M}$. According to the definition, one can see that $\text{adj}(\textbf{M})$ is an integer matrix, if $\textbf{M}$ is an integer matrix. Throughout this paper, all matrices are square matrices, unless otherwise stated.

\section{Preliminaries}\label{sec2}

To make this paper self-contained, this section reviews some of formal definitions and basic properties pertaining to lattices, integer vectors, and integer matrices \cite{lattice1,chent,md-crt}.

1) \textit{Lattice}: Given a $D\times D$ nonsingular matrix $\textbf{M}\in\mathbb{R}^{D\times D}$, a lattice generated by $\textbf{M}$ is defined as
\begin{equation}
\mathcal{L}(\textbf{M})=\left\{\textbf{M}\textbf{n}\,|\,\textbf{n}\in\mathbb{Z}^{D}\right\}.
\end{equation}

2) \textit{The shortest vector problem (SVP) on lattice}: For a lattice $\mathcal{L}(\textbf{M})$ that is generated by a nonsingular matrix $\textbf{M}\in\mathbb{R}^{D\times D}$, its minimum distance, denoted as $\lambda_{\mathcal{L}(\textbf{M})}$, is defined as the smallest distance between any two distinct lattice points, i.e.,
\begin{equation}\label{svp}
\lambda_{\mathcal{L}(\textbf{M})}=\min_{\substack{\textbf{w},\textbf{v}\in\mathcal{L}(\textbf{M}), \\\textbf{w}\neq\textbf{v}}}\lVert\textbf{w}-\textbf{v}\rVert.
\end{equation}
As we know, a lattice is closed under addition and subtraction operations. The minimum distance of $\mathcal{L}(\textbf{M})$ is therefore equal to the length (magnitude) of the shortest non-zero lattice point, i.e., $\lambda_{\mathcal{L}(\textbf{M})}=\min_{\textbf{v}\in\mathcal{L}(\textbf{M})\backslash\{\textbf{0}\}}\,\lVert\textbf{v}\rVert$.

3) \textit{The closest vector problem (CVP) on lattice}: For a lattice $\mathcal{L}(\textbf{M})$ that is generated by a nonsingular matrix $\textbf{M}\in\mathbb{R}^{D\times D}$, the closest lattice point in $\mathcal{L}(\textbf{M})$ to a given arbitrary point $\textbf{w}\in\mathbb{R}^{D}$ is defined as
\begin{equation}\label{cvp}
\textbf{p}=\argmin_{\textbf{v}\in\mathcal{L}(\textbf{M})}\,\lVert\textbf{v}-\textbf{w}\rVert.
\end{equation}

\textit{Remark}: There have been many algorithms for handling the SVP and CVP problems in the literature (see, e.g., \cite{latticproblem,latticproblem2}).
We note that the distance above in (\ref{svp}) and (\ref{cvp}) can be measured by any norm of vectors, such as the $\ell_2$ norm $\lVert\textbf{v}\rVert_2=\sqrt{\sum_i\lvert v(i)\rvert^2}$,
the $\ell_1$ norm $\lVert\textbf{v}\rVert_1=\sum_i\lvert v(i)\rvert$, and the $\ell_\infty$ norm $\lVert\textbf{v}\rVert_\infty=\max_i\lvert v(i)\rvert$.
In this paper, the SVP and CVP problems are identified as the integer quadratic programming problems. Hence, we can solve them (i.e., (\ref{svp}) and (\ref{cvp})) utilizing enumeration \cite{enum} and MOSEK with CVX \cite{boyd}, respectively.

4) \textit{Notation} $\mathcal{N}(\textbf{M})$: Given a $D\times D$ nonsingular integer matrix $\textbf{M}\in\mathbb{Z}^{D\times D}$, the notation $\mathcal{N}(\textbf{M})$ is defined as
\begin{equation}
\mathcal{N}(\textbf{M})=\left\{\textbf{k}\,|\,\textbf{k}=\textbf{M}\textbf{x}, \textbf{x}\in[0,1)^{D} \text{ and } \textbf{k}\in\mathbb{Z}^{D}\right\}.
\end{equation}
The number of elements in $\mathcal{N}(\textbf{M})$ is equal to $\left|\text{det}(\textbf{M})\right|$.

5) \textit{Division representation for integer vectors}: Given a $D\times D$ nonsingular integer matrix $\textbf{M}\in\mathbb{Z}^{D\times D}$, any integer vector $\textbf{m}\in\mathbb{Z}^{D}$ can be uniquely represented as $\textbf{m}=\textbf{M}\textbf{n}+\textbf{r}$ with $\textbf{r}\in\mathcal{N}(\textbf{M})$ and $\textbf{n}\in\mathbb{Z}^{D}$. For modular representation, it is denoted as
\begin{equation}
\textbf{m}\equiv \textbf{r} \!\!\mod \textbf{M},
\end{equation}
where $\textbf{M}$ is a modulus, and $\textbf{n}$ and $\textbf{r}$ are the folding vector and remainder of $\textbf{m}$ with respect to $\textbf{M}$, respectively.

\textit{Remark}: The folding vector and the remainder are computed as $\textbf{n}=\lfloor\textbf{M}^{-1}\textbf{m}\rfloor$ and $\textbf{r}=\textbf{m}-\textbf{M}\lfloor\textbf{M}^{-1}\textbf{m}\rfloor$. As $\lfloor\textbf{M}^{-1}\textbf{m}\rfloor$ may suffer from round-off errors due to finite precision on computers, an alternative for computing $\textbf{r}$ is given by
\begin{equation}\label{remaindercal}
\textbf{r}=\textbf{M}\left(\text{adj}(\textbf{M})\textbf{m}\!\!\mod \text{det}(\textbf{M})\right)/\text{det}(\textbf{M}),
\end{equation}
in which the operation ``mod'' means that $\text{adj}(\textbf{M})\textbf{m}$ is element-wisely modulo $\text{det}(\textbf{M})$.

6) \textit{Unimodular matrix}: A square matrix $\textbf{U}$ is unimodular if it is an integer matrix with $|\text{det}(\textbf{U})|=1$. For a unimodular matrix $\textbf{U}$, its inverse $\textbf{U}^{-1}$ is unimodular, due to $\textbf{U}^{-1}=\text{adj}(\textbf{U})/\text{det}(\textbf{U})$.

7) \textit{Divisor}: An integer matrix $\textbf{A}$ is a left divisor of an integer matrix $\textbf{M}$ if $\textbf{A}^{-1}\textbf{M}$ is an integer matrix. If $\textbf{A}$ is a left divisor of each of all $L\geq2$ integer matrices $\textbf{M}_1, \textbf{M}_2, \cdots, \textbf{M}_L$, we call $\textbf{A}$ a common left divisor (cld) of $\textbf{M}_1, \textbf{M}_2, \cdots, \textbf{M}_L$. In particular, $\textbf{A}$ is a greatest common left divisor (gcld) of $\textbf{M}_1, \textbf{M}_2, \cdots, \textbf{M}_L$, if any other cld is a left divisor of $\textbf{A}$. One can readily see that among all cld's, a gcld has the largest absolute determinant and is unique up to post-multiplication by a unimodular matrix.

8) \textit{Multiple}: A nonsingular integer matrix $\textbf{A}$ is a left multiple of an integer matrix $\textbf{M}$, if there is a nonsingular integer matrix $\textbf{P}$ such that $\textbf{A}=\textbf{P}\textbf{M}$. When $\textbf{A}$ is a left multiple of each of all $L\geq2$ integer matrices $\textbf{M}_1, \textbf{M}_2, \cdots, \textbf{M}_L$, we call $\textbf{A}$ a common left multiple (clm) of $\textbf{M}_1, \textbf{M}_2, \cdots, \textbf{M}_L$. In particular, $\textbf{A}$ is a least common left multiple (lclm) of $\textbf{M}_1, \textbf{M}_2, \cdots, \textbf{M}_L$, if any other clm is a left multiple of $\textbf{A}$. Apparently, among all clm's, an lclm has the smallest absolute determinant and is unique up to pre-multiplication by a unimodular matrix.

\textit{Remark}: Similar to 5) and 6) above, we can define right divisor/multiple, common right divisor/multiple (crd/crm), greatest common right divisor (gcrd), and least common right multiple (lcrm), respectively. Both divisors and multiples are supposed to be nonsingular integer matrices throughout this paper.

9) \textit{Coprimeness}: Two integer matrices $\textbf{M}$ and $\textbf{N}$ are said to be left (right) coprime if their gcld (gcrd) is unimodular. If $\textbf{M}$ and $\textbf{N}$ are commutative, i.e., $\textbf{M}\textbf{N}=\textbf{N}\textbf{M}$, their left coprimeness and right coprimeness imply each other, and so we will simply use ``coprimeness''. If $\textbf{M}$ and $\textbf{N}$ are commutative and coprime, $\textbf{M}\textbf{N}$ is an lcrm and an lclm, and so we will simply use ``lcm''.

10) \textit{Bezout's theorem}: Let $\textbf{L}\in\mathbb{Z}^{D\times D}$ be a gcld of two integer matrices $\textbf{M}$ and $\textbf{N}\in\mathbb{Z}^{D\times D}$. There exist integer matrices $\textbf{P}$ and $\textbf{Q}\in\mathbb{Z}^{D\times D}$ such that
\begin{equation}\label{bezo}
\textbf{M}\textbf{P}+\textbf{N}\textbf{Q}=\textbf{L}.
\end{equation}
Of note, how to compute the accompanying matrices $\textbf{P}$ and $\textbf{Q}$ will be presented in 12) below.
Similarly, letting $\textbf{L}\in\mathbb{Z}^{D\times D}$ be a gcrd of $\textbf{M}$ and $\textbf{N}$, there exist integer matrices $\textbf{P}$ and $\textbf{Q}$ such that $\textbf{P}\textbf{M}+\textbf{Q}\textbf{N}=\textbf{L}$.

11) \textit{The Smith form}: A rank-$\gamma$ integer matrix $\textbf{M}\in\mathbb{Z}^{D\times K}$ can be factorized as
\begin{equation}\label{smithform}
\textbf{U}\textbf{M}\textbf{V}=
\begin{cases}
    \left(
      \begin{array}{cc}
        \bm{\Lambda} & \bm{0} \\
      \end{array}
    \right), & \text{if } K>D,\\
        \bm{\Lambda},
       & \text{if } K=D,\\
    \left(
      \begin{array}{c}
        \bm{\Lambda} \\
        \bm{0} \\
      \end{array}
    \right),
                 & \text{if } K<D,
\end{cases}
\end{equation}
where $\textbf{U}\in\mathbb{Z}^{D\times D}$ and $\textbf{V}\in\mathbb{Z}^{K\times K}$ are unimodular matrices,  and $\bm{\Lambda}$ is a $\min(K,D)\times\min(K,D)$ diagonal integer matrix, i.e., $\bm{\Lambda}\triangleq\text{diag}(\delta_1,\delta_2,\cdots,\delta_\gamma,0,\cdots,0)$. If we suppose that $\delta_1,\delta_2,\cdots,\delta_\gamma$ are positive and $\delta_i$ divides $\delta_{i+1}$ for each $1\leq i\leq\gamma-1$, then $\bm{\Lambda}$ is unique for the given matrix $\textbf{M}$, while $\textbf{U}$ and $\textbf{V}$ are generally not.
In addition, $\delta_1,\delta_2,\cdots,\delta_\gamma$ are termed the invariant factors and can be obtained by $\delta_i=d_i/d_{i-1}$ for $1\leq i\leq \gamma$, where $d_i$ is the gcd of all $i\times i$ determinantal minors of $\textbf{M}$ and $d_0=1$.

12) \textit{Calculation of gcld}: To compute a gcld of two nonsingular integer matrices $\textbf{M}$ and $\textbf{N}\in\mathbb{Z}^{D\times D}$, we let $\textbf{H}=\left(
             \begin{array}{cc}
               \textbf{M} & \textbf{N} \\
             \end{array}\right)\in\mathbb{Z}^{D\times 2D}$
and obtain the Smith form $\textbf{U}\left(
             \begin{array}{cc}
               \textbf{M} & \textbf{N} \\
             \end{array}\right)\textbf{V}=\left(
             \begin{array}{cc}
               \bm{\Lambda} & \bm{0} \\
             \end{array}\right)$,
where $\textbf{U}\in\mathbb{Z}^{D\times D}$ and $\textbf{V}\in\mathbb{Z}^{2D\times 2D}$ are unimodular matrices, and $\bm{\Lambda}\in\mathbb{Z}^{D\times D}$ is a diagonal integer matrix (which is also nonsingu-\\lar due to $\text{rank}(\textbf{H})=D$). After simple computations, we obtain $\left(
             \begin{array}{cc}
               \textbf{M} & \textbf{N} \\
             \end{array}\right)=\left(
             \begin{array}{cc}
               \textbf{L} & \bm{0} \\
             \end{array}\right)\textbf{V}^{-1}$,
where $\textbf{L}=\textbf{U}^{-1}\bm{\Lambda}$. Since $\textbf{U}^{-1}$ is unimodular, $\textbf{L}$ is a nonsingular integer matrix, i.e., $\textbf{L}\in\mathbb{Z}^{D\times D}$. Since $\textbf{V}^{-1}$ is unimodular, we can partition $\textbf{V}^{-1}$ into four $D\times D$ integer matrix blocks $\textbf{K}_{ij}\in\mathbb{Z}^{D\times D}$ for $1\leq i,j\leq2$, and obtain
\begin{equation}
\left(\begin{array}{cc}
\textbf{M} & \textbf{N} \\
\end{array}\right)=\left(
\begin{array}{cc}
\textbf{L} & \bm{0} \\
\end{array}\right)\left(
                    \begin{array}{cc}
                      \textbf{K}_{11} & \textbf{K}_{12} \\
                      \textbf{K}_{21} & \textbf{K}_{22}\\
                    \end{array}
                  \right).
\end{equation}
We therefore have $\textbf{M}=\textbf{L}\textbf{K}_{11}$ and $\textbf{N}=\textbf{L}\textbf{K}_{12}$. It is proved that such $\textbf{L}$ is actually a gcld of $\textbf{M}$ and $\textbf{N}$ (see \cite{md-crt} for the proof).

\textit{Remark}: We then provide a way to compute the accompanying matrices $\textbf{P}$ and $\textbf{Q}$ in (\ref{bezo}) for the Bezout's theorem. From the Smith form of $\textbf{H}$ above, we get $\left(
             \begin{array}{cc}
               \textbf{M} & \textbf{N} \\
             \end{array}\right)\textbf{V}=\left(
             \begin{array}{cc}
               \textbf{L} & \bm{0} \\
             \end{array}\right)$.
We partition $\textbf{V}$ into four $D\times D$ integer matrix blocks $\textbf{V}_{ij}\in\mathbb{Z}^{D\times D}$ for $1\leq i,j\leq2$, and have
\begin{equation}
\left(\begin{array}{cc}
\textbf{M} & \textbf{N} \\
\end{array}\right)\left(
                                \begin{array}{cc}
                                  \textbf{V}_{11} & \textbf{V}_{12} \\
                                  \textbf{V}_{21} & \textbf{V}_{22} \\
                                \end{array}
                              \right)=\left(
\begin{array}{cc}
\textbf{L} & \bm{0} \\
\end{array}\right).
\end{equation}
It implies the Bezout's theorem, expressed by $\textbf{M}\textbf{V}_{11}+\textbf{N}\textbf{V}_{21}=\textbf{L}$, i.e., $\textbf{P}=\textbf{V}_{11}$ and $\textbf{Q}=\textbf{V}_{21}$ in (\ref{bezo}).

13) \textit{Calculation of lcrm}: To calculate an lcrm of two nonsingular integer matrices $\textbf{M}$ and $\textbf{N}\in\mathbb{Z}^{D\times D}$, we let $\textbf{H}=\textbf{M}^{-1}\textbf{N}$. Because of $\textbf{M}^{-1}=\text{adj}(\textbf{M})/\text{det}(\textbf{M})$, $\textbf{M}^{-1}$ has all elements being rational numbers, and so does $\textbf{H}$. Letting $d$ be the lcm of the denominators of all elements in $\textbf{H}$, we know that $d\textbf{H}$ is a $D\times D$ nonsingular integer matrix. We compute the Smith form of $d\textbf{H}$ as $\textbf{U}d\textbf{H}\textbf{V}=\bm{\Lambda}$, i.e.,
\begin{equation}\label{yuel}
\textbf{M}^{-1}\textbf{N}=\textbf{U}^{-1}\text{diag}(\delta_1/d,\delta_2/d,\cdots,\delta_D/d)\textbf{V}^{-1},
\end{equation}
where $\textbf{U}$ and $\textbf{V}$ are $D\times D$ unimodular matrices (i.e., $\textbf{U},\textbf{V}\in\mathbb{Z}^{D\times D}$), and $\bm{\Lambda}=\text{diag}(\delta_1,\delta_2,\cdots,\delta_D)\in\mathbb{Z}^{D\times D}$ as derived in (\ref{smithform}). All the rational numbers $\delta_1/d,\delta_2/d,\cdots,\delta_D/d$ are represented by their irreducible forms; that is to say, for $1\leq i\leq D$, $\delta_i/d=\alpha_i/\beta_i$ where $\alpha_i$ and $\beta_i$ are coprime positive integers. Let $\bm{\Lambda}_\alpha=\text{diag}(\alpha_1,\alpha_2,\cdots,\alpha_D)$ and $\bm{\Lambda}_\beta=\text{diag}(\beta_1,\beta_2,\cdots,\beta_D)$. Based on (\ref{yuel}), we obtain $\textbf{M}^{-1}\textbf{N}=\textbf{U}^{-1}\bm{\Lambda}_\alpha\bm{\Lambda}_\beta^{-1}\textbf{V}^{-1}$.
Let $\textbf{P}=\textbf{U}^{-1}\bm{\Lambda}_\alpha$ and $\textbf{Q}=\textbf{V}\bm{\Lambda}_\beta$, which are clearly nonsingular integer matrices and right coprime. We hence have $\textbf{M}^{-1}\textbf{N}=\textbf{P}\textbf{Q}^{-1}$,
i.e., $\textbf{M}\textbf{P}=\textbf{N}\textbf{Q}$. It is proved that $\textbf{R}\triangleq\textbf{M}\textbf{P}=\textbf{N}\textbf{Q}$ is actually an lcrm of $\textbf{M}$ and $\textbf{N}$ (see \cite{chent} for the proof).

\textit{Remark}: For $L\geq3$ nonsingular integer matrices $\left\{\textbf{M}_i\right\}_{i=1}^L$, we can compute an lcrm of $\left\{\textbf{M}_i\right\}_{i=1}^L$ via computing an lcrm of two matrices iteratively, due to the fact that $\text{lcrm}\left(\textbf{M}_1,\textbf{M}_2,\cdots,\textbf{M}_L\right)$
$=\text{lcrm}\left(\text{lcrm}\left(\textbf{M}_1,\textbf{M}_2,\cdots,\textbf{M}_{L-1}\right),\textbf{M}_L\right)$ holds, which has been proved in \cite{md-crt}. Besides,
similar to the calculations of gcld and lcrm above, the calculations of gcrd and lclm can be obtained. For more details, we refer the reader to \cite{chent,md-crt}.

\section{Previous Results on (Robust) MD-CRT}\label{sec3}
Consider a system of congruences
\begin{equation}\label{ddon}
\left\{\begin{array}{ll}
\textbf{m}\equiv \textbf{r}_1  \!\!\mod \textbf{M}_1\\
\textbf{m}\equiv \textbf{r}_2  \!\!\mod \textbf{M}_2\\
\:\;\;\;\;\vdots\\
\textbf{m}\equiv \textbf{r}_L  \!\!\mod \textbf{M}_L,\\
\end{array}\right.
\end{equation}
where moduli $\left\{\textbf{M}_i\right\}_{i=1}^{L}\in\mathbb{Z}^{D\times D}$ are nonsingular integer matrices, and
$\textbf{R}\in\mathbb{Z}^{D\times D}$ is anyone of their lcrm's. With respect to (\ref{ddon}), let us recall the results about the (robust) MD-CRT we recently proposed in \cite{md-crt} as follows. For simplicity of notation, we will use $\textbf{r}=\langle\textbf{m}\rangle_{\textbf{M}}$ to denote the remainder $\textbf{r}$ of $\textbf{m}$ modulo $\textbf{M}$.

\subsection{MD-CRT}
\begin{proposition}[\!\!\!\cite{md-crt}\!]\label{pro1}
Let moduli $\left\{\textbf{M}_i\right\}_{i=1}^{L}$ in (\ref{ddon}) be arbitrary nonsingular integer matrices. An integer vector $\textbf{m}\in\mathcal{N}(\textbf{R})$ can be accurately reconstructed from its remainders $\{\textbf{r}_i\}_{i=1}^L$.
\end{proposition}

Notice that a cascaded reconstruction algorithm for the MD-CRT in Proposition \ref{pro1} is introduced in \cite{md-crt}. For $2\leq i\leq L$, let $\textbf{R}_{i}$ be an lcrm of $\left\{\textbf{M}_k\right\}_{k=1}^{i-1}$, $\textbf{G}_{i}$ be a gcld of $\textbf{M}_i$ and $\textbf{R}_{i}$, and $\textbf{P}_{i}$ and $\textbf{Q}_{i}$ be the accompanying matrices in the Bezout's theorem with $\textbf{R}_{i}\textbf{P}_{i}+\textbf{M}_{i}\textbf{Q}_{i}=\textbf{G}_{i}$. On the basis of 12) and 13) in Sec. \ref{sec2}, all these involved matrices can be computed in advance.
Here, we briefly summarize the core steps of the cascaded reconstru-\\ction algorithm for the MD-CRT.

\begin{itemize}
  \item A solution (denoted as $\textbf{m}_1\in\mathcal{N}(\textbf{R}_{3})$) to
     \begin{equation}\label{firststage}
    \left\{\begin{array}{ll}
    \textbf{m}\equiv \textbf{r}_1 \!\!\mod \textbf{M}_1\\
    \textbf{m}\equiv \textbf{r}_2 \!\!\mod \textbf{M}_2\\
    \end{array}\right.
    \end{equation}
    is obtained as $\textbf{m}_1=\left\langle\textbf{r}_1+\textbf{M}_1\textbf{P}_{2}\textbf{G}_{2}^{-1}(\textbf{r}_2-\textbf{r}_1)\right\rangle_{\textbf{R}_{3}}$.
  \item Based on the cascade architecture of the congruences, we next obtain a solution (denoted as $\textbf{m}_2\in\mathcal{N}(\textbf{R}_{4})$) to
  \begin{equation}\label{secondstage}
    \left\{\begin{array}{ll}
    \textbf{m}\equiv \textbf{m}_1 \!\!\mod \textbf{R}_{3}\\
    \textbf{m}\equiv \textbf{r}_3 \!\!\mod \textbf{M}_3\\
    \end{array}\right.
    \end{equation}
    as $\textbf{m}_2=\left\langle\textbf{m}_1+\textbf{R}_3\textbf{P}_{3}\textbf{G}_{3}^{-1}(\textbf{r}_3-\textbf{m}_1)\right\rangle_{\textbf{R}_{4}}$.
  \item Following the above steps, we assemble two congruences at a time, until a solution (denoted as $\textbf{m}_{L-1}\in\mathcal{N}(\textbf{R})$) to
        \begin{equation}
    \left\{\begin{array}{ll}
    \textbf{m}\equiv \textbf{m}_{L-2} \!\!\mod \textbf{R}_{L}\\
    \textbf{m}\equiv \textbf{r}_L \!\!\mod \textbf{M}_L\\
    \end{array}\right.
    \end{equation}
is calculated as $\textbf{m}_{L-1}=\left\langle\textbf{m}_{L-2}+\textbf{R}_L\textbf{P}_{L}\textbf{G}_{L}^{-1}(\textbf{r}_L-\textbf{m}_{L-2})\right\rangle_{\textbf{R}}$.
As verified in \cite{md-crt}, the lcrm (i.e., $\textbf{R}$) of $\left\{\textbf{M}_i\right\}_{i=1}^L$ is an lcrm of $\textbf{R}_{L}$ and $\textbf{M}_L$, and $\textbf{m}_{L-1}$ is a unique solution to (\ref{ddon}) from the MD-CRT if $\textbf{m}\in\mathcal{N}(\textbf{R})$, i.e., $\textbf{m}=\textbf{m}_{L-1}$.
\end{itemize}

\textit{Remark}: If the moduli $\left\{\textbf{M}_i\right\}_{i=1}^{L}\in\mathbb{Z}^{D\times D}$ are pairwise commutative and coprime,  it is clear that $\textbf{R}=\textbf{M}_1\textbf{M}_2\cdots\textbf{M}_L\textbf{U}\in\mathbb{Z}^{D\times D}$ is an lcrm of all the moduli for any unimodular matrix $\textbf{U}$, and the MD-CRT in Proposition \ref{pro1} has a closed-form solution as
\begin{equation}
\textbf{m}=\left\langle\sum_{i=1}^{L}\textbf{W}_i\widehat{\textbf{W}}_i\textbf{r}_i\right\rangle_{\textbf{R}},
\end{equation}
where $\textbf{W}_i=\textbf{M}_1\cdots\textbf{M}_{i-1}\textbf{M}_{i+1}\cdots\textbf{M}_L$, and $\widehat{\textbf{W}}_i$ is the accompanying matrix in the Bezout's theorem ($\textbf{W}_i\widehat{\textbf{W}}_i+\textbf{M}_i\textbf{Q}_i=\textbf{I}$
with $\textbf{Q}_i\in\mathbb{Z}^{D\times D}$) and can be calculated in advance.

\subsection{Robust MD-CRT for a special class of moduli}
In \cite{md-crt}, the robust MD-CRT was first proposed for a special class of moduli, i.e., moduli $\left\{\textbf{M}_i\right\}_{i=1}^{L}$ in (\ref{ddon}) are given by
\begin{equation}\label{specmod}
\textbf{M}_i=\textbf{M}\bm{\Gamma}_i\,\text{ for }1\leq i\leq L,
\end{equation}
where $\left\{\bm{\Gamma}_i\right\}_{i=1}^{L}\in\mathbb{Z}^{D\times D}$ are pairwise commutative and coprime, and $\textbf{M}\in\mathbb{Z}^{D\times D}$. In this special case, $\textbf{R}=\textbf{M}\bm{\Gamma}_1\bm{\Gamma}_2\cdots\bm{\Gamma}_L\textbf{U}$ for any unimodular matrix $\textbf{U}$ is an lcrm of $\left\{\textbf{M}_i\right\}_{i=1}^L$, and the basic idea of the robust MD-CRT in \cite{md-crt} is to accurately determine the folding vectors $\{\textbf{n}_i\}_{i=1}^{L}$ from the erroneous remainders
\begin{equation}
\tilde{\textbf{r}}_i\triangleq\textbf{r}_i+\triangle\textbf{r}_i\in\mathcal{N}(\textbf{M}_i)\,\text{ for }1\leq i\leq L,
\end{equation}
and afterwards obtain a robust reconstruction of $\textbf{m}$ as
\begin{equation}\label{rrr}
\tilde{\textbf{m}}=\frac{1}{L}\sum_{i=1}^{L}\left(\textbf{M}_{i}\textbf{n}_{i}+\tilde{\textbf{r}}_{i}\right),
\end{equation}
where $\left\{\triangle\textbf{r}_i\right\}_{i=1}^{L}\in\mathbb{Z}^{D}$ are the remainder errors.
Define
\begin{equation}\label{rangess}
\mathcal{A}_i\triangleq\left\{\textbf{m}\in\mathbb{Z}^{D}\,|\,\, \lfloor \textbf{M}_i^{-1}\textbf{m}\rfloor\in\mathcal{N}(\bm{\Gamma}_1\cdots\bm{\Gamma}_{i-1}\bm{\Gamma}_{i+1}\cdots\bm{\Gamma}_L\textbf{U}_i)\right\}
\end{equation}
for $1\leq i\leq L$, where $\left\{\textbf{U}_i\right\}_{i=1}^L$ are any unimodular matrices. The robust MD-CRT for this special class of moduli expressed in (\ref{specmod}) was obtained in \cite{md-crt}, as stated below.

\begin{proposition}[\!\!\!\cite{md-crt}\!]\label{pro2}
Let moduli $\left\{\textbf{M}_i\right\}_{i=1}^{L}$ in (\ref{ddon}) be given by (\ref{specmod}). We can accurately determine the folding vectors $\{\textbf{n}_i\}_{i=1}^{L}$ of an integer vector $\textbf{m}\in\bigcup_{i=1}^{L}\mathcal{A}_i$ (without loss of generality, we suppose that $\textbf{m}\in\mathcal{A}_{1}$) from the erroneous remainders $\left\{\tilde{\textbf{r}}_i\right\}_{i=1}^L$, if and only if
\begin{equation}\label{xiaxia}
\textbf{0}=\argmin_{\textbf{h}\in\mathcal{L}(\textbf{M})}\,\lVert\textbf{h}-(\triangle\textbf{r}_i-\triangle\textbf{r}_1)\rVert\; \text{ for }2\leq i\leq L.
\end{equation}
Moreover, letting $\tau$ be the remainder error bound, i.e., $\lVert\triangle\textbf{r}_i\rVert\leq\tau$ for $1\leq i\leq L$, a simple sufficient condition is
\begin{equation}\label{condition_suf2}
\tau<\frac{\lambda_{\mathcal{L}(\textbf{M})}}{4}.
\end{equation}
Once $\{\textbf{n}_i\}_{i=1}^L$ are accurately determined, we can obtain a robust reconstruction $\tilde{\textbf{m}}$ of $\textbf{m}$ by (\ref{rrr}) such that $\lVert\tilde{\textbf{m}}-\textbf{m}\rVert\leq\tau$.
\end{proposition}

The necessary and sufficient condition (\ref{xiaxia}) means that the lattice point $\textbf{0}$ in $\mathcal{L}(\textbf{M})$ is the only closest lattice point to the difference of the remainder errors $\triangle\textbf{r}_i$ and $\triangle\textbf{r}_1$ for every $i$, $2\leq$\\$i\leq L$.

\textit{Remark}: In \cite{md-crt}, a closed-form reconstruction algorithm for the robust MD-CRT in Proposition \ref{pro2} was also provided.

\section{Robust MD-CRT for General Moduli}\label{sec4}
When moduli do not satisfy the constraint in (\ref{specmod}), the results (i.e., Proposition \ref{pro2} above) and reconstruction algorithm in \cite{md-crt} cannot be directly applied, which might limit the applications of the robust MD-CRT in practice. In this section, we consider the robust MD-CRT for a general set of moduli on which the constraint imposed in \cite{md-crt} is no longer required.

We can equivalently write (\ref{ddon}) as
\begin{equation}\label{ddon11}
\left\{\begin{array}{ll}
\textbf{m}=\textbf{M}_1\textbf{n}_1+\textbf{r}_1 \\
\textbf{m}=\textbf{M}_2\textbf{n}_2+\textbf{r}_2\\
\:\;\;\;\;\vdots\\
\textbf{m}=\textbf{M}_L\textbf{n}_L+\textbf{r}_L,\\
\end{array}\right.
\end{equation}
where $\{\textbf{n}_i\}_{i=1}^L$ are the folding vectors.
Without loss of generality, letting the first equation in (\ref{ddon11}) be a reference, we subtract it from the last $L-1$ equations, i.e.,
\begin{equation}\label{ddon22}
\left\{\begin{array}{ll}
\textbf{M}_1\textbf{n}_1-\textbf{M}_2\textbf{n}_2=\textbf{r}_2-\textbf{r}_1 \\
\textbf{M}_1\textbf{n}_1-\textbf{M}_3\textbf{n}_3=\textbf{r}_3-\textbf{r}_1\\
\:\;\;\;\;\:\;\;\;\;\:\;\;\;\;\:\;\;\;\;\:\;\;\vdots\\
\textbf{M}_1\textbf{n}_1-\textbf{M}_L\textbf{n}_L=\textbf{r}_L-\textbf{r}_1.\\
\end{array}\right.
\end{equation}
Define $\textbf{M}_{1i}=\text{gcld}(\textbf{M}_1,\textbf{M}_i)$, ${\bf{\Gamma}}_{1i}=\textbf{M}_{1i}^{-1}\textbf{M}_1$, and ${\bf{\Gamma}}_{i1}=\textbf{M}_{1i}^{-1}\textbf{M}_i$ for $2\leq i\leq L$. Then, left-multiplying $\textbf{M}_{1i}^{-1}$ on both sides of the $(i-1)$-th equation in (\ref{ddon22}) for $2\leq i\leq L$, we get
\begin{equation}\label{ddon33}
\left\{\begin{array}{ll}
{\bf{\Gamma}}_{12}\textbf{n}_1-{\bf{\Gamma}}_{21}\textbf{n}_2=\textbf{M}_{12}^{-1}(\textbf{r}_2-\textbf{r}_1)\vspace{1ex} \\
{\bf{\Gamma}}_{13}\textbf{n}_1-{\bf{\Gamma}}_{31}\textbf{n}_3=\textbf{M}_{13}^{-1}(\textbf{r}_3-\textbf{r}_1)\\
\:\;\;\;\;\:\;\;\;\;\:\;\;\;\;\:\;\;\;\;\:\;\;\vdots\\
{\bf{\Gamma}}_{1L}\textbf{n}_1-{\bf{\Gamma}}_{L1}\textbf{n}_L=\textbf{M}_{1L}^{-1}(\textbf{r}_L-\textbf{r}_1).\\
\end{array}\right.
\end{equation}
From (\ref{ddon33}), we know that $\left\{\textbf{M}_{1i}^{-1}(\textbf{r}_i-\textbf{r}_1)\right\}_{i=2}^{L}$ are integer vectors, i.e., for $2\leq i\leq L$,
\begin{equation}
\textbf{r}_i-\textbf{r}_1\in \mathcal{L}(\textbf{M}_{1i}).
\end{equation}

In the same way as that used in \cite{md-crt}, for each $2\leq i\leq L$, we estimate $\textbf{r}_i-\textbf{r}_1$ from the erroneous remainders $\left\{\tilde{\textbf{r}}_i\right\}_{i=1}^{L}$ through finding a closest lattice point $\textbf{v}_i$ in $\mathcal{L}(\textbf{M}_{1i})$ to $\tilde{\textbf{r}}_i-\tilde{\textbf{r}}_1$, i.e.,
\begin{equation}\label{xxiao}
\textbf{v}_i=\argmin_{\textbf{v}\in\mathcal{L}(\textbf{M}_{1i})}\,\lVert\textbf{v}-(\tilde{\textbf{r}}_i-\tilde{\textbf{r}}_1)\rVert.
\end{equation}
Instead of accurately determining the folding vectors $\left\{\textbf{n}_i\right\}_{i=1}^L$ in \cite{md-crt}, we intend to accurately determine $\left\{\textbf{M}_i\textbf{n}_i\right\}_{i=1}^L$. Specifically, by taking the modulo-$\textbf{M}_i$ on both sides of the $(i-1)$-th equation in (\ref{ddon22}) for $2\leq i\leq L$, we have
\begin{equation}\label{haohao}
\left\{\begin{array}{ll}
\textbf{M}_1\textbf{n}_1\equiv \textbf{0}  \!\!\mod \textbf{M}_1\\
\textbf{M}_1\textbf{n}_1\equiv \textbf{r}_2-\textbf{r}_1  \!\!\mod \textbf{M}_2\\
\textbf{M}_1\textbf{n}_1\equiv \textbf{r}_3-\textbf{r}_1  \!\!\mod \textbf{M}_3\\
\:\;\;\;\;\;\;\;\;\;\vdots\\
\textbf{M}_1\textbf{n}_1\equiv \textbf{r}_L-\textbf{r}_1  \!\!\mod \textbf{M}_L,\\
\end{array}\right.
\end{equation}
where the first equation spontaneously holds. Once $\left\{\textbf{r}_i-\textbf{r}_1\right\}_{i=2}^{L}$ are accurately estimated from (\ref{xxiao}), i.e., $\textbf{v}_i=\textbf{r}_i-\textbf{r}_1$ for $2\leq i\leq L$, we can accurately determine $\textbf{M}_1\textbf{n}_1$ from (\ref{haohao}) according to the MD-CRT (see Proposition \ref{pro1} above), provided that $\textbf{M}_1\textbf{n}_1\in$
$\mathcal{N}\left(\text{lcrm}(\textbf{M}_1,\textbf{M}_2,\cdots,\textbf{M}_L)\right)$, equivalently written as $\lfloor\textbf{M}_1^{-1}\textbf{m}\rfloor\in$ $\mathcal{N}\left(\textbf{M}_1^{-1}\text{lcrm}(\textbf{M}_1,\textbf{M}_2,\cdots,\textbf{M}_L)\right)$.
Then, $\textbf{M}_i\textbf{n}_i$ can be accurately determined as $\textbf{M}_1\textbf{n}_1-\textbf{v}_i$ for each $2\leq i\leq L$. In this end, we derive the following lemma, which can be proved similarly to Theorem 3 in \cite{md-crt}.

\begin{lemma}\label{lema1}
Let moduli $\left\{\textbf{M}_i\right\}_{i=1}^{L}$ in (\ref{ddon}) be $L$ distinct arbitrary nonsingular integer matrices, and an integer vector $\textbf{m}$ be within the range
\begin{equation}\label{rrange}
\lfloor\textbf{M}_1^{-1}\textbf{m}\rfloor\in\mathcal{N}\left(\textbf{M}_1^{-1}\text{lcrm}(\textbf{M}_1,\textbf{M}_2,\cdots,\textbf{M}_L)\right).
\end{equation}
We can accurately determine $\{\textbf{M}_i\textbf{n}_i\}_{i=1}^{L}$ from the erroneous re-\\mainders $\left\{\tilde{\textbf{r}}_i\right\}_{i=1}^L$, if and only if
\begin{equation}\label{sn}
\textbf{0}=\argmin_{\textbf{h}\in\mathcal{L}(\textbf{M}_{1i})}\,\lVert\textbf{h}-(\triangle\textbf{r}_i-\triangle\textbf{r}_1)\rVert\; \text{ for }2\leq i\leq L.
\end{equation}
Moreover, letting $\tau$ be the remainder error bound, i.e., $\lVert\triangle\textbf{r}_i\rVert\leq\tau$ for $1\leq i\leq L$, a simple sufficient condition is
\begin{equation}\label{condition_suf2}
\tau<\min\limits_{2\leq i\leq L}\frac{\lambda_{\mathcal{L}(\textbf{M}_{1i})}}{4}.
\end{equation}
After $\{\textbf{M}_i\textbf{n}_i\}_{i=1}^L$ are accurately determined, a robust reconstruction $\tilde{\textbf{m}}$ of $\textbf{m}$ can be obtained by (\ref{rrr}) with $\lVert\tilde{\textbf{m}}-\textbf{m}\rVert\leq\tau$.
\end{lemma}
\begin{proof}
From (\ref{xxiao}), we have, for $2\leq i\leq L$,
\begin{equation}\label{xxiaoxx}
\textbf{v}_i=\argmin_{\textbf{v}\in\mathcal{L}(\textbf{M}_{1i})}\,\lVert\textbf{v}-(\textbf{r}_i-\textbf{r}_1)-(\triangle\textbf{r}_i-\triangle\textbf{r}_1)\rVert.
\end{equation}
Due to $\textbf{v}\in\mathcal{L}(\textbf{M}_{1i})$ and $\textbf{r}_i-\textbf{r}_1\in\mathcal{L}(\textbf{M}_{1i})$, we have $\textbf{v}-(\textbf{r}_i-\textbf{r}_1)\in\mathcal{L}(\textbf{M}_{1i})$, and (\ref{xxiaoxx}) can be equivalently written as
\begin{equation}\label{xxiaoxx2}
\textbf{h}_i=\argmin_{\textbf{h}\in\mathcal{L}(\textbf{M}_{1i})}\,\lVert\textbf{h}-(\triangle\textbf{r}_i-\triangle\textbf{r}_1)\rVert
\end{equation}
by taking $\textbf{h}=\textbf{v}-(\textbf{r}_i-\textbf{r}_1)$.

We first prove the sufficiency of (\ref{sn}). If $\textbf{h}_i=\textbf{0}$ for $2\leq i\leq L$, we get $\textbf{v}_i=\textbf{r}_i-\textbf{r}_1$, i.e., $\left\{\textbf{r}_i-\textbf{r}_1\right\}_{i=2}^{L}$ are accurately obtained from (\ref{xxiao}). Hence, as mentioned before, $\{\textbf{M}_i\textbf{n}_i\}_{i=1}^{L}$ can be accurately determined, when (\ref{rrange}) satisfies.

We next prove the necessity of (\ref{sn}). Assume that there exists at least one $\textbf{h}_{k_0}$ that does not satisfy (\ref{sn}), i.e., $\textbf{h}_{k_0}\neq\textbf{0}$, for some $k_0$ with $2\leq k_0\leq L$. Furthermore, due to $\textbf{v}_{k_0}=\textbf{h}_{k_0}+(\textbf{r}_{k_0}-\textbf{r}_1)$, we know $\textbf{v}_{k_0}\neq\textbf{r}_{k_0}-\textbf{r}_1$. We then have the following two cases.

\textit{Case A:} $\textbf{h}_{l_0}\notin\mathcal{L}(\textbf{M}_{l_0})$ for some $l_0$ with $2\leq l_0\leq L$ (where $l_0$ is not necessarily equal to $k_0$), i.e., $\textbf{h}_{l_0}\neq\textbf{M}_{l_0}\textbf{n}$ for any $\textbf{n}\in\mathbb{Z}^D$. In this case, it is ready to see that $\textbf{v}_{l_0}$ and $\textbf{r}_{l_0}-\textbf{r}_1$ have different remainders modulo $\textbf{M}_{l_0}$. Thus, according to the uniqueness of the reconstruction in the MD-CRT, $\textbf{M}_1\textbf{n}_1$ cannot be accurately determined from $\left\{\textbf{v}_i\right\}_{i=1}^L$ in (\ref{haohao}).

\textit{Case B:} For each $2\leq i\leq L$, $\textbf{h}_{i}\in\mathcal{L}(\textbf{M}_{i})$, i.e., $\textbf{h}_{i}=\textbf{M}_{i}\textbf{n}$ for some $\textbf{n}\in\mathbb{Z}^D$. In this case,
considering that $\textbf{v}_{i}=\textbf{h}_{i}+(\textbf{r}_{i}-\textbf{r}_1)$, we know that $\textbf{v}_{i}$ and $\textbf{r}_{i}-\textbf{r}_1$ have the same remainders modulo $\textbf{M}_{i}$ for each $2\leq i\leq L$,
and therefore, $\textbf{M}_1\textbf{n}_1$ can be accurately determined from $\left\{\textbf{v}_i\right\}_{i=1}^L$ in (\ref{haohao}) using the MD-CRT. However, since $\textbf{v}_{k_0}\neq\textbf{r}_{k_0}-\textbf{r}_1$, the reconstruction of $\textbf{M}_{k_0}\textbf{n}_{k_0}$ as $\textbf{M}_1\textbf{n}_1-\textbf{v}_{k_0}$ is not accurate. This completes the proof of the necessity part.

Ultimately, we prove the simple sufficient condition in (\ref{condition_suf2}) for accurately determining $\{\textbf{M}_i\textbf{n}_i\}_{i=1}^{L}$. Assume that there exists one $\textbf{h}_{q_0}$ in (\ref{xxiaoxx2}) satisfying $\textbf{h}_{q_0}\neq\textbf{0}$ for some $q_0$ with $2\leq q_0\leq L$. We have
\begin{equation}
\begin{split}
\lVert\textbf{h}_{q_0}\rVert
&=\lVert\textbf{h}_{q_0}-(\triangle\textbf{r}_{q_0}-\triangle\textbf{r}_1)-(\textbf{0}-(\triangle\textbf{r}_{q_0}-\triangle\textbf{r}_1))\rVert\\
&\leq \lVert\textbf{h}_{q_0}-(\triangle\textbf{r}_{q_0}-\triangle\textbf{r}_1)\rVert+\lVert\triangle\textbf{r}_{q_0}-\triangle\textbf{r}_1\rVert \\
& \leq 2\lVert\triangle\textbf{r}_{q_0}-\triangle\textbf{r}_1\rVert \\
& \leq 4\tau < \lambda_{\mathcal{L}(\textbf{M}_{1{q_0}})}\,,
\end{split}
\end{equation}
in which the second inequality follows from the fact that $\textbf{h}_{q_0}$ is one closest lattice point in $\mathcal{L}(\textbf{M}_{1{q_0}})$ to $\triangle\textbf{r}_{q_0}-\triangle\textbf{r}_1$, and the last inequality holds since
$4\tau<\min_{2\leq i\leq L}\,\lambda_{\mathcal{L}(\textbf{M}_{1i})}\leq\lambda_{\mathcal{L}(\textbf{M}_{1{q_0}})}$. Hence, it contradicts with $\textbf{h}_{q_0}\in\mathcal{L}(\textbf{M}_{1{q_0}})$, i.e., $\lVert\textbf{h}_{q_0}\rVert\geq\lambda_{\mathcal{L}(\textbf{M}_{1{q_0}})}$, which indicates that the condition in (\ref{condition_suf2}) implies (\ref{sn}).

Once $\{\textbf{M}_i\textbf{n}_i\}_{i=1}^{L}$ are accurately determined, we have a robust reconstruction $\tilde{\textbf{m}}$ of $\textbf{m}$ as $\tilde{\textbf{m}}=\frac{1}{L}\sum_{i=1}^{L}\left(\textbf{M}_{i}\textbf{n}_{i}+\tilde{\textbf{r}}_{i}\right)$, i.e.,
\begin{equation}
\begin{split}
\lVert\tilde{\textbf{m}}-\textbf{m}\rVert& =\left\lVert\frac{1}{L}\sum_{i=1}^{L}\left(\textbf{M}_{i}\textbf{n}_{i}+\textbf{r}_{i}+\triangle\textbf{r}_i\right)-\textbf{m}\right\rVert\\
& =\left\lVert\frac{1}{L}\sum_{i=1}^{L}\triangle\textbf{r}_i\right\rVert\leq\frac{1}{L}\sum_{i=1}^{L}\left\lVert\triangle\textbf{r}_i\right\rVert\leq\tau.
\end{split}
\end{equation}
This completes the proof of the lemma.
\end{proof}

Note that in the aforementioned analysis, we just arbitrarily select the first equation (or the first remainder $\textbf{r}_1$) in (\ref{ddon11}) as a reference to be subtracted from the other equations to acquire (\ref{ddon22}), followed by Lemma \ref{lema1}. In fact, we can further improve the reconstruction robustness of the robust MD-CRT via selecting a proper reference equation in (\ref{ddon11}). Define $\textbf{M}_{ij}=\text{gcld}(\textbf{M}_i,\textbf{M}_j)$ for $1\leq i\neq j\leq L$. Find the index $l_0$ with $1\leq l_0\leq L$ such that
\begin{equation}\label{findd}
\min_{1\leq j\leq L \atop j\neq l_0}\lambda_{\mathcal{L}(\textbf{M}_{{l_0}j})}=\max_{1\leq i\leq L}\min_{1\leq j\leq L \atop j\neq i}\lambda_{\mathcal{L}(\textbf{M}_{ij})}.
\end{equation}
By treating the $l_0$-th remainder as the reference and following the above procedures utilized in Lemma \ref{lema1}, we obtain the result below straightforwardly, along with a closed-form reconstruction algorithm (see \textbf{Algorithm \ref{algo1}}) for the robust MD-CRT.

\begin{theorem}\label{theo1}
Let moduli $\left\{\textbf{M}_i\right\}_{i=1}^{L}$ in (\ref{ddon}) be $L$ different arbitrary nonsingular integer matrices.
Suppose that the index $l_0$ with $1\leq l_0\leq L$ satisfies (\ref{findd}). For an integer vector $\textbf{m}$ with
\begin{equation}\label{rrange11}
\lfloor\textbf{M}_{l_0}^{-1}\textbf{m}\rfloor\in\mathcal{N}\left(\textbf{M}_{l_0}^{-1}\text{lcrm}(\textbf{M}_1,\textbf{M}_2,\cdots,\textbf{M}_L)\right),
\end{equation}
we can accurately determine $\{\textbf{M}_i\textbf{n}_i\}_{i=1}^{L}$ from the erroneous rem-\\ainders $\left\{\tilde{\textbf{r}}_i\right\}_{i=1}^L$ by \textbf{Algorithm \ref{algo1}}, if and only if
\begin{equation}\label{sn11}
\textbf{0}=\argmin_{\textbf{h}\in\mathcal{L}(\textbf{M}_{{l_0}j})}\,\lVert\textbf{h}-(\triangle\textbf{r}_j-\triangle\textbf{r}_{l_0})\rVert\; \text{ for }1\leq j\leq L \text{ and }j\neq l_0.
\end{equation}
Moreover, letting $\tau$ be the remainder error bound, i.e., $\lVert\triangle\textbf{r}_i\rVert\leq\tau$ for $1\leq i\leq L$, a simple sufficient condition is
\begin{equation}\label{condition_suf211}
\tau<\max\limits_{1\leq i\leq L}\min\limits_{1\leq j\leq L \atop j\neq i}\frac{\lambda_{\mathcal{L}(\textbf{M}_{ij})}}{4}=\min\limits_{1\leq j\leq L \atop j\neq l_0}\frac{\lambda_{\mathcal{L}(\textbf{M}_{{l_0}j})}}{4}.
\end{equation}
After $\{\textbf{M}_i\textbf{n}_i\}_{i=1}^L$ are accurately determined, a robust reconstruction $\tilde{\textbf{m}}$ of $\textbf{m}$ can be obtained by (\ref{rrr}) with $\lVert\tilde{\textbf{m}}-\textbf{m}\rVert\leq\tau$.
\end{theorem}

\begin{algorithm}[!h]
\caption{}

\vspace{1mm}
\begin{algorithmic}[1]\label{algo1}
\STATE According to 12) in Sec. \ref{sec2}, calculate $\textbf{M}_{l_0 j}=\text{gcld}(\textbf{M}_{l_0},\textbf{M}_{j})$ for $1\leq j\leq L$ and $j\neq l_0$.

\STATE According to 13) in Sec. \ref{sec2}, calculate $\textbf{R}_3=\text{lcrm}(\textbf{M}_1,\textbf{M}_2)$, $\textbf{R}_4=\text{lcrm}(\textbf{M}_1,\textbf{M}_2,\textbf{M}_3)=\text{lcrm}(\textbf{R}_3,\textbf{M}_3)$,
$\textbf{R}_5=\text{lcrm}(\textbf{M}_1,$ $\textbf{M}_2,\textbf{M}_3,\textbf{M}_4)=\text{lcrm}(\textbf{R}_4,\textbf{M}_4)$, $\cdots\cdots$, $\textbf{R}=\textbf{R}_{L+1}=\text{lcrm}($ $\textbf{M}_1,\textbf{M}_2, \cdots,\textbf{M}_L)=\text{lcrm}(\textbf{R}_L,\textbf{M}_L)$.

\STATE According to 3) in Sec. \ref{sec2}, from the given $\left\{\tilde{\textbf{r}}_i\right\}_{i=1}^L$, calculate $\textbf{v}_j$ for $1\leq j\leq L$ and $j\neq l_0$ as
\begin{equation}\label{xxiao22}
\textbf{v}_j=\argmin_{\textbf{v}\in\mathcal{L}(\textbf{M}_{{l_0}j})}\,\lVert\textbf{v}-(\tilde{\textbf{r}}_j-\tilde{\textbf{r}}_{l_0})\rVert.
\end{equation}

\STATE Calculate $\textbf{M}_{l_0}\tilde{\textbf{n}}_{l_0}\in\mathcal{N}(\textbf{R})=\mathcal{N}(\text{lcrm}(\textbf{M}_1,\textbf{M}_2, \cdots,\textbf{M}_L))$ via the cascaded reconstruction
algorithm for the MD-CRT in Proposition \ref{pro1} from the following system of congruences
\begin{equation}\label{haohao22}
\left\{\begin{array}{ll}
\textbf{M}_{l_0}\tilde{\textbf{n}}_{l_0}\equiv \textbf{v}_1  \!\!\mod \textbf{M}_1\\
\:\;\;\;\;\;\;\;\;\;\vdots\\
\textbf{M}_{l_0}\tilde{\textbf{n}}_{l_0}\equiv \textbf{v}_{l_0-1}  \!\!\mod \textbf{M}_{l_0-1}\\
\textbf{M}_{l_0}\tilde{\textbf{n}}_{l_0}\equiv \textbf{0}  \!\!\mod \textbf{M}_{l_0}\\
\textbf{M}_{l_0}\tilde{\textbf{n}}_{l_0}\equiv \textbf{v}_{l_0+1}  \!\!\mod \textbf{M}_{l_0+1}\\
\:\;\;\;\;\;\;\;\;\;\vdots\\
\textbf{M}_{l_0}\tilde{\textbf{n}}_{l_0}\equiv \textbf{v}_L  \!\!\mod \textbf{M}_L.\\
\end{array}\right.
\end{equation}

\STATE Calculate $\textbf{M}_{j}\tilde{\textbf{n}}_{j}=\textbf{M}_{l_0}\tilde{\textbf{n}}_{l_0}-\textbf{v}_j$ for $1\leq j\leq L$ and $j\neq l_0$. Then, a reconstruction of $\textbf{m}$ is $\tilde{\textbf{m}}=\frac{1}{L}\sum_{i=1}^{L}(\textbf{M}_{i}\tilde{\textbf{n}}_{i}+\tilde{\textbf{r}}_{i})$.
\end{algorithmic}
\end{algorithm}

\textit{Remark}: When the moduli $\left\{\textbf{M}_i\right\}_{i=1}^{L}$ in Theorem \ref{theo1} satisfy the constraint (i.e., (\ref{specmod})) imposed in \cite{md-crt}, Theorem \ref{theo1} reduces to Proposition \ref{pro2}. It should also be pointed out that the MD-CRT reconstruction range $\textbf{m}\in\mathcal{N}\left(\text{lcrm}(\textbf{M}_1,\textbf{M}_2,\cdots,\textbf{M}_L)\right)$ and the robust MD-CRT reconstruction range in (\ref{rrange11}) do not imply ea-\\ch other, unless for the (robust) $1$-D CRT and the (robust) MD-CRT with moduli being nonsingular diagonal integer matrices. We take an example as follows. Let $\textbf{M}_1=\left(
                  \begin{array}{cc}
                    1 & 3 \\
                    3 & 1 \\
                  \end{array}
                \right)$ and $\textbf{M}_2=\left(
                  \begin{array}{cc}
                    1 & 2 \\
                    2 & 1 \\
                  \end{array}
                \right)$, whose product $\textbf{R}=\textbf{M}_1\textbf{M}_2=\left(
                  \begin{array}{cc}
                    7 & 5 \\
                    5 & 7 \\
                  \end{array}
                \right)$ is their lcrm. When $\textbf{m}=\left(
                  \begin{array}{c}
                    5  \\
                    4  \\
                  \end{array}
                \right)=\textbf{R}\left(
                  \begin{array}{c}
                    5/8  \\
                    1/8 \\
                  \end{array}
                \right)\in\mathcal{N}(\textbf{R})$, we obtain $\textbf{n}_1=\left(
                  \begin{array}{c}
                    0  \\
                    1 \\
                  \end{array}
                \right)$ \\$=\textbf{M}_2\left(
                  \begin{array}{c}
                    2/3  \\
                    -1/3 \\
                  \end{array}
                \right)$, indicating $\textbf{n}_1\notin\mathcal{N}(\textbf{M}_1^{-1}\textbf{R})=\mathcal{N}(\textbf{M}_2)$. On the other hand, when $\textbf{m}=\left(
                  \begin{array}{c}
                    10  \\
                    9  \\
                  \end{array}
                \right)=\textbf{R}\left(
                  \begin{array}{c}
                    25/24  \\
                    13/24 \\
                  \end{array}
                \right)\notin\mathcal{N}(\textbf{R})$, we get
                $\textbf{n}_1=\left(
                  \begin{array}{c}
                    2  \\
                    2 \\
                  \end{array}
                \right)=\textbf{M}_2\left(
                  \begin{array}{c}
                    2/3  \\
                    2/3 \\
                  \end{array}
                \right)$, implying $\textbf{n}_1\in\mathcal{N}(\textbf{M}_1^{-1}\textbf{R})=\mathcal{N}(\textbf{M}_2)$.
Owing to this reconstruction range inequivalence, we cannot obtain a further improved variant of the robust MD-CRT as in  \cite{xiaoxia1}, where a multi-stage (e.g., second-stage) robust 1-D CRT was generalized by first splitting the congruences into several groups, then applying the robust 1-D CRT to each group inde-\\
pendently, and finally applying the robust 1-D CRT again to a new system of congruences with the reconstructions and lcm's in all the groups being the remainders and moduli, respectively.

For a better understanding of Theorem \ref{theo1}, we next present an example to explain our implementation of the robust MD-CRT through the step-by-step procedures in \textbf{Algorithm \ref{algo1}}.

\begin{example}\label{exam1}
Consider $L=3$ moduli $\textbf{M}_1=\left(
                  \begin{array}{cc}
                    5850 & 9000 \\
                    2580 & 2940 \\
                  \end{array}
                \right)$, $\textbf{M}_2=\left(
                  \begin{array}{cc}
                    28950 & 24150 \\
                    14140 & 11680 \\
                  \end{array}
                \right)$, and $\textbf{M}_3=\left(
                  \begin{array}{cc}
                    3440 & 3460 \\
                    1540 & 1160 \\
                  \end{array}
                \right)$. Let $\textbf{m}=\left(
                  \begin{array}{c}
                    -5365350  \\
                    -2402280  \\
                  \end{array}
                \right)$, then the remainders of $\textbf{m}$ modulo $\left\{\textbf{M}_i\right\}_{i=1}^3$ can be calculated from (\ref{remaindercal}) as $\textbf{r}_1=\left(
                  \begin{array}{c}
                    0  \\
                    0  \\
                  \end{array}
                \right)$, $\textbf{r}_2=\left(
                  \begin{array}{c}
                    37650  \\
                    18320  \\
                  \end{array}
                \right)$, and  $\textbf{r}_3=\left(
                  \begin{array}{c}
                    4490  \\
                    1660  \\
                  \end{array}
                \right)$. Correspondingly, the folding vectors are given by $\textbf{n}_1=\left(
                  \begin{array}{c}
                    -971  \\
                    35  \\
                  \end{array}
                \right)$, $\textbf{n}_2=\left(
                  \begin{array}{c}
                    1390  \\
                    -1890  \\
                  \end{array}
                \right)$, and $\textbf{n}_3=\left(
                  \begin{array}{c}
                    -1561 \\
                    0  \\
                  \end{array}
                \right)$. Let the erroneous remainders be $\tilde{\textbf{r}}_1=\left(
                  \begin{array}{c}
                    52 \\
                    36  \\
                  \end{array}
                \right)$, $\tilde{\textbf{r}}_2=\left(
                  \begin{array}{c}
                    37673 \\
                    18243  \\
                  \end{array}
                \right)$, and $\tilde{\textbf{r}}_3=\left(
                  \begin{array}{c}
                    4446 \\
                    1610  \\
                  \end{array}
                \right)$, with their respective remainder errors $\triangle\textbf{r}_1=\left(
                  \begin{array}{c}
                    52 \\
                    36  \\
                  \end{array}
                \right)$, $\triangle\textbf{r}_2=\left(
                  \begin{array}{c}
                    23 \\
                    -77  \\
                  \end{array}
                \right)$, and $\triangle\textbf{r}_3=\left(
                  \begin{array}{c}
                    -44 \\
                    -50  \\
                  \end{array}
                \right)$. In the following, we elaborate how to robustly reconstruct $\textbf{m}$ from the erroneous remainders $\left\{\tilde{\textbf{r}}_i\right\}_{i=1}^3$ by \textbf{Algorithm \ref{algo1}}.

$\romannumeral1)$ First, calculate $\textbf{M}_{12}=\left(
                  \begin{array}{cc}
                    -2272650 & -2274600 \\
                    -1002640 & -1003500 \\
                  \end{array}
                \right)$, $\textbf{M}_{13}=$
                $\left(
                  \begin{array}{cc}
                    -604610 & -454920 \\
                    -266740 & -200700 \\
                  \end{array}
                \right)$,
                 $\textbf{M}_{23}=\left(
                  \begin{array}{cc}
                    -3632710 & -3661660 \\
                    -1774320 & -1788460 \\
                  \end{array}
                \right)$, according to 12) in Sec. \ref{sec2}. Under the $\ell_2$ norm, we then obtain $\lambda_{\mathcal{L}(\textbf{M}_{12})}=637.89$, $\lambda_{\mathcal{L}(\textbf{M}_{13})}=352.28$,  $\lambda_{\mathcal{L}(\textbf{M}_{23})}=178.04$. Finally, from (\ref{findd}), we regard the first remainder as the reference, and the reconstruction robustness bound is $352.28/4=88.07$. One can easily see that
                 the remainder error bound $\tau$ satisfies $\lVert\triangle\textbf{r}_i\rVert\leq\tau<88.07$ for $1\leq i\leq3$.

$\romannumeral2)$ According to 13) in Sec. \ref{sec2}, calculate $\textbf{R}_3=\text{lcrm}(\textbf{M}_1,\textbf{M}_2)$
$=\left(
                  \begin{array}{cc}
                    86850 & -101250 \\
                    42420 & -49800 \\
                  \end{array}
                \right)$,
                followed by $\textbf{R}=\text{lcrm}(
                \textbf{M}_1,\textbf{M}_2,\textbf{M}_3)$
                $=\text{lcrm}(\textbf{R}_3,\textbf{M}_3)
                =
                \left(
                  \begin{array}{cc}
                    774000 & -6133500 \\
                    346500 & -2746200 \\
                  \end{array}
                \right)$. In addition, based on 12) in Sec. \ref{sec2}, we calculate the accompanying matrices $\textbf{P}_2$
                $=\left(
                  \begin{array}{cc}
                    -10 & -12 \\
                    -69 & -69 \\
                  \end{array}
                \right)$ and $\textbf{Q}_2=\left(
                  \begin{array}{cc}
                    -25 & -28 \\
                    -36 & -32 \\
                  \end{array}
                \right)$ satisfying $\textbf{M}_1\textbf{P}_2+\textbf{M}_2\textbf{Q}_2=\textbf{M}_{12}$, and calculate the accompanying matrices $\textbf{P}_3$
                $=\left(
                  \begin{array}{cc}
                    -108 & -65 \\
                    85 & 51 \\
                  \end{array}
                \right)$ and $\textbf{Q}_3=\left(
                  \begin{array}{cc}
                    -40 & -24 \\
                    -45 & -27 \\
                  \end{array}
                \right)$ satisfying $\textbf{R}_3\textbf{P}_3+\textbf{M}_3\textbf{Q}_3=\textbf{G}_3\triangleq\text{gcld}(\textbf{R}_3,\textbf{M}_3)=\left(
                  \begin{array}{cc}
                    -18279350 & -10984980 \\
                    -8928160 & -5365380 \\
                  \end{array}
                \right)$.

$\romannumeral3)$ According to 3) in Sec. \ref{sec2}, calculate $\textbf{v}_2$ and $\textbf{v}_3$ from (\ref{xxiao22}) as $\textbf{v}_2=\left(
                  \begin{array}{c}
                    37650 \\
                    18320  \\
                  \end{array}
                \right)$ and $\textbf{v}_3=\left(
                  \begin{array}{c}
                    4490 \\
                    1660  \\
                  \end{array}
                \right)$. One can easily confirm that $\lfloor\textbf{M}_{1}^{-1}\textbf{m}\rfloor=\textbf{n}_1\in\mathcal{N}\left(\textbf{M}_{1}^{-1}\textbf{R}\right)=\mathcal{N}\left(\left(
                  \begin{array}{cc}
                    140 & -1110 \\
                    -5 & 40 \\
                  \end{array}
                \right)\right)$, i.e., $\left(
                  \begin{array}{c}
                    -971  \\
                    35  \\
                  \end{array}
                \right)=\left(
                  \begin{array}{cc}
                    140 & -1110 \\
                    -5 & 40 \\
                  \end{array}
                \right)\left(
                  \begin{array}{c}
                    0.2  \\
                    0.9  \\
                  \end{array}
                \right)$, and $\lVert\triangle\textbf{r}_i\rVert\leq\tau<88.07$ for $1\leq i\leq 3$. Therefore, Theorem \ref{theo1} holds.

$\romannumeral4)$ Via the cascaded reconstruction algorithm for the MD-CRT in Proposition \ref{pro1}, calculate $\bm{\zeta}\triangleq\textbf{M}_1\tilde{\textbf{n}}_1$ from
\begin{equation}
\left\{\begin{array}{ll}
\bm{\zeta}\equiv \textbf{0}  \!\!\mod \textbf{M}_1\\
\bm{\zeta}\equiv \textbf{v}_2  \!\!\mod \textbf{M}_2\\
\bm{\zeta}\equiv \textbf{v}_3 \!\!\mod \textbf{M}_3.\\
\end{array}\right.
\end{equation}
\begin{itemize}
  \item From (\ref{firststage}), we acquire $\bm{\zeta}_1=\left\langle\textbf{0}+\textbf{M}_1\textbf{P}_{2}\textbf{M}_{12}^{-1}(\textbf{v}_2-\textbf{0})\right\rangle_{\textbf{R}_{3}}= \left(
                  \begin{array}{c}
                    -20250  \\
                    -9960  \\
                  \end{array}
                \right)$.
  \item From (\ref{secondstage}), we get
  $\bm{\zeta}=\bm{\zeta}_2=\left\langle\bm{\zeta}_1+\textbf{R}_3\textbf{P}_{3}\textbf{G}_{3}^{-1}(\textbf{v}_3-\bm{\zeta}_1)\right\rangle_{\textbf{R}}=\left(
                  \begin{array}{c}
                    -5365350  \\
                    -2402280  \\
                  \end{array}
                \right)$.
\end{itemize}
We so have $\tilde{\textbf{n}}_1=\left(
                  \begin{array}{c}
                    -971  \\
                    35  \\
                  \end{array}
                \right)$, which is equal to $\textbf{n}_1$.

$\romannumeral5)$ Calculate $\textbf{M}_{2}\tilde{\textbf{n}}_{2}=\textbf{M}_{1}\tilde{\textbf{n}}_{1}-\textbf{v}_2=\left(
                  \begin{array}{c}
                    -5403000  \\
                    -2420600  \\
                  \end{array}
                \right)$ as well as $\textbf{M}_{3}\tilde{\textbf{n}}_{3}=\textbf{M}_{1}\tilde{\textbf{n}}_{1}-\textbf{v}_3=\left(
                  \begin{array}{c}
                    -5369840  \\
                    -2403940  \\
                  \end{array}
                \right)$. It also implies that $\tilde{\textbf{n}}_{2}=$ $\left(
                  \begin{array}{c}
                    1390  \\
                    -1890  \\
                  \end{array}
                \right)$ and $\tilde{\textbf{n}}_{3}=\left(
                  \begin{array}{c}
                    -1561  \\
                    0 \\
                  \end{array}
                \right)$, which are equal to $\textbf{n}_2$ and $\textbf{n}_3$, respectively.
               Therefore, $\left\{\textbf{M}_i\textbf{n}_i\right\}_{i=1}^3$ (i.e., $\left\{\textbf{n}_i\right\}_{i=1}^3$) are accurately determined from the erroneous remainders $\left\{\tilde{\textbf{r}}_i\right\}_{i=1}^3$, and a robust reconstruction of $\textbf{m}$ can be obtained as $\tilde{\textbf{m}}=\frac{1}{3}\sum_{i=1}^{3}(\textbf{M}_{i}\tilde{\textbf{n}}_{i}+\tilde{\textbf{r}}_{i})=\left(
                  \begin{array}{c}
                    -5365339.67  \\
                    -2402310.33 \\
                  \end{array}
                \right)$, i.e., $\lVert\tilde{\textbf{m}}-\textbf{m}\rVert_2=32.05\leq\tau<88.07$. \hfill $\blacksquare$
\end{example}

Since in the above new results in Theorem \ref{theo1} there is no any constraint on moduli $\left\{\textbf{M}_i\right\}_{i=1}^{L}$ (i.e., moduli $\left\{\textbf{M}_i\right\}_{i=1}^{L}$ are arbitrary nonsingular integer matrices), some of these moduli might be redundant with respect to the reconstruction robustness bound (i.e., (\ref{condition_suf211})), while retaining the reconstruction range (i.e., (\ref{rrange11})). We investigate the case when there exists a pair of moduli $\textbf{M}_{i_1}$ and $\textbf{M}_{i_2}$ such that $\textbf{M}_{i_1}$ $=\textbf{M}_{i_2}\textbf{P}$ for $\textbf{P}\in\mathbb{Z}^{D\times D}$, i.e., $\textbf{M}_{i_1}$ is a right multiple of $\textbf{M}_{i_2}$. For this, we have the following corollary.

\begin{corollary}\label{cor2}
If there are two moduli $\textbf{M}_{i_1}$ and $\textbf{M}_{i_2}$ in $\left\{\textbf{M}_i\right\}_{i=1}^{L}$ in Theorem \ref{theo1} such that $\textbf{M}_{i_1}$ $=\textbf{M}_{i_2}\textbf{P}$ for $\textbf{P}\in\mathbb{Z}^{D\times D}$, the modulus $\textbf{M}_{i_2}$ is redundant, in the sense that the appearance of $\textbf{M}_{i_2}$ does not help increase (and might even decrease) the reconstruction robustness bound,
meanwhile keeping the reconstruction range unchanged. As such, $\textbf{M}_{i_2}$ can be deleted from the set of moduli in this case.
\end{corollary}
\begin{proof}
Without loss of generality, let us assume that $\textbf{M}_{1}=\textbf{M}_{L}\textbf{P}$ for $\textbf{P}\in\mathbb{Z}^{D\times D}$. We first prove $\lambda_{\mathcal{L}(\textbf{M}_{1j})}\geq\lambda_{\mathcal{L}(\textbf{M}_{Lj})}$ for any $2\leq j\leq L-1$. Since $\textbf{M}_{Lj}=\text{gcld}(\textbf{M}_{L},\textbf{M}_{j})$ and $\textbf{M}_{1}=\textbf{M}_{L}\textbf{P}$ for any $2\leq j\leq L-1$, it is ready to confirm that $\textbf{M}_{Lj}$ is a cld of $\textbf{M}_{1}$ and $\textbf{M}_{j}$. Therefore, $\textbf{M}_{Lj}$ is a left divisor of $\textbf{M}_{1j}$ from the definition of gcld, i.e., $\textbf{M}_{1j}=\textbf{M}_{Lj}\textbf{Q}_j$ for $\textbf{Q}_j\in\mathbb{Z}^{D\times D}$. That is to say, $\mathcal{L}(\textbf{M}_{1j})\subseteq\mathcal{L}(\textbf{M}_{Lj})$, and so $\lambda_{\mathcal{L}(\textbf{M}_{1j})}\geq\lambda_{\mathcal{L}(\textbf{M}_{Lj})}$.

For the set of moduli $\left\{\textbf{M}_i\right\}_{i=1}^{L-1}$, let $s$ denote the reconstruction robustness bound, i.e., $s=\max_{1\leq i\leq L-1}\min_{1\leq j\leq L-1 \atop j\neq i}\lambda_{\mathcal{L}(\textbf{M}_{ij})}/4$. For the set of moduli $\left\{\textbf{M}_i\right\}_{i=1}^{L}$, the reconstruction robustness bound can be expressed as
\begin{equation}\label{qiang}
\max_{1\leq i\leq L}\min_{1\leq j\leq L \atop j\neq i}\frac{\lambda_{\mathcal{L}(\textbf{M}_{ij})}}{4}=\max\bigg\{\underbrace{\max_{1\leq i\leq L-1}\min_{1\leq j\leq L \atop j\neq i}\frac{\lambda_{\mathcal{L}(\textbf{M}_{ij})}}{4}}_{(a)}, \underbrace{\min_{1\leq j\leq L-1}\frac{\lambda_{\mathcal{L}(\textbf{M}_{Lj})}}{4}}_{(b)}\bigg\}.
\end{equation}
As for $(a)$, due to $\min_{1\leq j\leq L \atop j\neq i}\lambda_{\mathcal{L}(\textbf{M}_{ij})}/4\leq\min_{1\leq j\leq L-1 \atop j\neq i}\lambda_{\mathcal{L}(\textbf{M}_{ij})}/4$, we have $\max_{1\leq i\leq L-1}\min_{1\leq j\leq L \atop j\neq i}\lambda_{\mathcal{L}(\textbf{M}_{ij})}/4\leq s$. As for $(b)$, since it has been proved above that $\lambda_{\mathcal{L}(\textbf{M}_{1j})} \geq\lambda_{\mathcal{L}(\textbf{M}_{Lj})}$ for any $2\leq j\leq L-1$, we have
\begin{equation}
\min_{1\leq j\leq L-1}\frac{\lambda_{\mathcal{L}(\textbf{M}_{Lj})}}{4}\leq\min_{2\leq j\leq L-1}\frac{\lambda_{\mathcal{L}(\textbf{M}_{Lj})}}{4}\leq\min_{2\leq j\leq L-1}\frac{\lambda_{\mathcal{L}(\textbf{M}_{1j})}}{4}\leq s.
\end{equation}
Thus, from (\ref{qiang}), we get $\max_{1\leq i\leq L}\min_{1\leq j\leq L \atop j\neq i}\lambda_{\mathcal{L}(\textbf{M}_{ij})}/4\leq s$, which suggests that the appearance of $\textbf{M}_{L}$ does not help increase the reconstruction robustness bound and might even worsen it.

For the set of moduli $\left\{\textbf{M}_i\right\}_{i=1}^{L}$, it is straightforward that $\textbf{M}_L$ is impossible to be a reference modulus (i.e., $l_0\neq L$ in Theorem \ref{theo1}), on account of $\lambda_{\mathcal{L}(\textbf{M}_{1j})}\geq\lambda_{\mathcal{L}(\textbf{M}_{Lj})}$ for any $2\leq j\leq L-1$. So, for the set of moduli $\left\{\textbf{M}_i\right\}_{i=1}^{L-1}$, we can choose the same $\textbf{M}_{l_0}$ as the reference modulus. Furthermore, owing to $\textbf{M}_1=\textbf{M}_L\textbf{P}$, we get $\text{lcrm}(\textbf{M}_1,\textbf{M}_2,\cdots,\textbf{M}_L)=\text{lcrm}(\textbf{M}_1,\textbf{M}_2,\cdots,\textbf{M}_{L-1})$, which implies from (\ref{rrange11}) that the reconstruction range remains uncha-
nged after deleting $\textbf{M}_L$ from moduli $\left\{\textbf{M}_i\right\}_{i=1}^{L}$.
\end{proof}

Going back to the necessary and sufficient condition in (\ref{sn11}) for the robust MD-CRT in Theorem \ref{theo1}, one can readily see that the remainder error difference bound depends on $\lambda_{\mathcal{L}(\textbf{M}_{{l_0}j})}$, i.e.,
\begin{equation}
\lVert\triangle\textbf{r}_j-\triangle\textbf{r}_{l_0}\rVert<\frac{\lambda_{\mathcal{L}(\textbf{M}_{{l_0}j})}}{2},
\end{equation}
for $1\leq j\leq L$ and $j\neq l_0$. It means that if we let $\tau_i$ denote the remainder error bound for the $i$-th remainder, i.e., $\lVert\triangle\textbf{r}_i\rVert\leq\tau_i$, for $1\leq i\leq L$, then $\{\tau_i\}_{i=1}^L$ will have different requirements for the robust reconstruction of $\textbf{m}$ in (\ref{rrange11}), as stated below.

\begin{corollary}\label{cor3}
Let moduli $\left\{\textbf{M}_i\right\}_{i=1}^{L}$ in (\ref{ddon}) be $L$ different arbitrary nonsingular integer matrices,
the index $l_0$ with $1\leq l_0\leq L$ satisfy (\ref{findd}), and an integer vector $\textbf{m}$ be with (\ref{rrange11}), as the same as those in Theorem \ref{theo1}. Let $\tau_i$ denote the remainder error bound for the $i$-th remainder, i.e., $\lVert\triangle\textbf{r}_i\rVert\leq\tau_i$, for $1\leq i\leq L$, among which the remainder error bound $\tau_{l_0}$ for the reference modulus $\textbf{M}_{l_0}$ is given by $\tau_{l_0}<\min_{1\leq j\leq L \atop j\neq l_0}\,\lambda_{\mathcal{L}(\textbf{M}_{{l_0}j})}/4$. If the remainder error bound $\tau_i$ for $1\leq i\leq L$ and $i\neq l_0$ satisfies
\begin{equation}
\lVert\triangle\textbf{r}_i\rVert\leq\tau_i\leq\frac{\lambda_{\mathcal{L}(\textbf{M}_{{l_0}i})}}{2}-\min_{1\leq j\leq L \atop j\neq l_0}\frac{\lambda_{\mathcal{L}(\textbf{M}_{{l_0}j})}}{4},
\end{equation}
we can accurately determine $\{\textbf{M}_i\textbf{n}_i\}_{i=1}^{L}$ from the erroneous rem-\\ainders $\left\{\tilde{\textbf{r}}_i\right\}_{i=1}^L$ by \textbf{Algorithm \ref{algo1}}, and therefore, a robust reconstruction $\tilde{\textbf{m}}$ of $\textbf{m}$ is obtained by (\ref{rrr}), i.e., $\lVert\tilde{\textbf{m}}-\textbf{m}\rVert\leq\sum_{i=1}^{L}\tau_i/L$.
\end{corollary}
\begin{proof}
As $\lVert\triangle\textbf{r}_{l_0}\rVert\leq\tau_{l_0}<\min_{1\leq j\leq L \atop j\neq l_0}\,\lambda_{\mathcal{L}(\textbf{M}_{{l_0}j})}/4$ and $\lVert\triangle\textbf{r}_i\rVert\leq\tau_i\leq\lambda_{\mathcal{L}(\textbf{M}_{{l_0}i})}/2-\min_{1\leq j\leq L \atop j\neq l_0}\lambda_{\mathcal{L}(\textbf{M}_{{l_0}j})}/4$ for $1\leq i\leq L$ and $i\neq l_0$, we have
\begin{equation}
\lVert\triangle\textbf{r}_j-\triangle\textbf{r}_{l_0}\rVert\leq\lVert\triangle\textbf{r}_j\rVert+\lVert\triangle\textbf{r}_{l_0}\rVert\leq\tau_{l_0}+\tau_i<\frac{\lambda_{\mathcal{L}(\textbf{M}_{{l_0}j})}}{2},
\end{equation}
which indicates (\ref{sn11}) in Theorem \ref{theo1}. As a result, $\{\textbf{M}_i\textbf{n}_i\}_{i=1}^{L}$ can be accurately determined from the erroneous remainders $\left\{\tilde{\textbf{r}}_i\right\}_{i=1}^L$ by \textbf{Algorithm \ref{algo1}}, and we can obtain a robust reconstruction $\tilde{\textbf{m}}$ of $\textbf{m}$ as $\tilde{\textbf{m}}=\frac{1}{L}\sum_{i=1}^{L}\left(\textbf{M}_{i}\textbf{n}_{i}+\tilde{\textbf{r}}_{i}\right)$, i.e.,
\begin{equation}
\lVert\tilde{\textbf{m}}-\textbf{m}\rVert =\left\lVert\frac{1}{L}\sum_{i=1}^{L}\triangle\textbf{r}_i\right\rVert\leq\frac{1}{L}\sum_{i=1}^{L}\left\lVert\triangle\textbf{r}_i\right\rVert\leq\frac{1}{L}\sum_{i=1}^{L}\tau_i.
\end{equation}
Therefore, Corollary \ref{cor3} is proved.
\end{proof}

\textit{Remark}: Of note, owing to $\lambda_{\mathcal{L}(\textbf{M}_{{l_0}i})}/2-\min_{1\leq j\leq L \atop j\neq l_0}\lambda_{\mathcal{L}(\textbf{M}_{{l_0}j})}/4\geq\min_{1\leq j\leq L \atop j\neq l_0}\lambda_{\mathcal{L}(\textbf{M}_{{l_0}j})}/4$ for $1\leq i\leq L$ and $i\neq l_0$, the allowed remain-\\der error bounds we derived by approaching them individually as above are larger than or equal to that in (\ref{condition_suf211}) for all the rema-\\inder errors in Theorem \ref{theo1}, while the reconstruction range (i.e., (\ref{rrange11})) remains unchanged. In addition, note that the counterpart results of Corollary \ref{cor2} and Corollary \ref{cor3} were also obtained for the robust 1-D CRT in \cite{xiaoxia1}.

\begin{example}\label{exam2}
Let us consider the $L=3$ moduli as in Example \ref{exam1}. According to Corollary \ref{cor3}, for the robust MD-CRT, we obtain the remainder error bounds as $\tau_1<352.28/4$, $\tau_2\leq923.5/4, \tau_3$\\ $\leq352.28/4$. One can obviously see that the allowed remainder error bounds here are larger than or equal to $352.28/4$ obtained in Theorem \ref{theo1}. Moreover, the reconstruction range in Corollary \ref{cor3} is the same as that (i.e., (\ref{rrange11})) in Theorem \ref{theo1}.
\hfill $\blacksquare$
\end{example}

\section{Generalization of Robust MD-CRT from Integer Vectors/Matrices to Real Ones}\label{sec5}
The above studies are all for integer vectors/matrices. Considering that in practical applications, an unknown vector (e.g., the phase of interest in multi-dimensional phase unwrapping in MIMO radar systems) is real-valued in general, we next ge-\\neralize the robust MD-CRT results in Theorem \ref{theo1} from integer vectors/matrices to real ones in this section. Note that we adopt boldfaced \textsf{Sans-Serif} letters to denote real vectors/matrices for distinguishing them from integer vectors/matrices.

Let $\textbf{\textsf{m}}$ be a $D$-dimensional real vector (i.e., $\textbf{\textsf{m}}\in\mathbb{R}^{D}$), which can be uniquely expressed as
\begin{equation}\label{bbb}
\textbf{\textsf{m}}=\textbf{\textsf{M}}\bm{\Psi}_i\textbf{n}_i+\textbf{\textsf{r}}_i\,\text{ for }1\leq i\leq L,
\end{equation}
where $\left\{\bm{\Psi}_i\right\}_{i=1}^L\in\mathbb{Z}^{D\times D}$ are known nonsingular integer matrices, $\textbf{\textsf{M}}\in\mathbb{R}^{D\times D}$ is a known nonsingular real matrix, and $\left\{\textbf{n}_i\right\}_{i=1}^L\in\mathbb{Z}^D$ are unknown integer vectors (or folding vectors). In particular,
$\left\{\textbf{\textsf{r}}_i\right\}_{i=1}^L\in\mathbb{R}^D$ are real vectors with $\textbf{\textsf{r}}_i\in\mathcal{F}(\textbf{\textsf{M}}\bm{\Psi}_i)$ for each $1\leq i\leq L$, which are real-valued versions of the previously mentioned integer remainders $\left\{\textbf{r}_i\right\}_{i=1}^L$ in (\ref{ddon11}). Here, $\mathcal{F}(\textbf{\textsf{M}}\bm{\Psi}_i)$ is termed the fundamental parallelepiped of $\mathcal{L}(\textbf{\textsf{M}}\bm{\Psi}_i)$, defined as
\begin{equation}
\mathcal{F}(\textbf{\textsf{M}}\bm{\Psi}_i)=\left\{\textbf{\textsf{M}}\bm{\Psi}_i\textbf{\textsf{x}}\,|\,\textbf{\textsf{x}}\in[0,1)^{D}\right\}.
\end{equation}
The volume of $\mathcal{F}(\textbf{\textsf{M}}\bm{\Psi}_i)$ equals $|\text{det}(\textbf{\textsf{M}}\bm{\Psi}_i)|$ \cite{PPV4}. $\mathcal{F}(\textbf{\textsf{M}}\bm{\Psi}_i)$ does not comprise any other lattice points in $\mathcal{L}(\textbf{\textsf{M}}\bm{\Psi}_i)$, except for the origin $\textbf{0}$. One can easily see that $\mathcal{F}(\textbf{\textsf{M}}\bm{\Psi}_i)$ and its shifted copies (i.e., $\mathcal{F}(\textbf{\textsf{M}}\bm{\Psi}_i)+\textbf{\textsf{v}}$ for any nonzero $\textbf{\textsf{v}}\in\mathcal{L}(\textbf{\textsf{M}}\bm{\Psi}_i)$) constitute the whole real vector space $\mathbb{R}^D$.

Let us define $\bm{\Psi}_{ij}=\text{gcld}(\bm{\Psi}_i,\bm{\Psi}_j)$ for $1\leq i\neq j\leq L$. Without loss of generality, we assume that $\bm{\Psi}_1$ satisfies
\begin{equation}
\min_{2\leq j\leq L}\lambda_{\mathcal{L}(\textbf{\textsf{M}}\bm{\Psi}_{1j})}=\max_{1\leq i\leq L}\min_{1\leq j\leq L \atop j\neq i}\lambda_{\mathcal{L}(\textbf{\textsf{M}}\bm{\Psi}_{ij})}.
\end{equation}
By treating $\textbf{\textsf{M}}\bm{\Psi}_1$ as the reference and following the operations used in (\ref{ddon22}) and (\ref{ddon33}), we have, from (\ref{bbb}),
\begin{equation}\label{cccd}
\left\{\begin{array}{ll}
\bm{\Psi}_1\textbf{n}_1-\bm{\Psi}_2\textbf{n}_2=\textbf{\textsf{M}}^{-1}(\textbf{\textsf{r}}_2-\textbf{\textsf{r}}_1)\vspace{0.5ex} \\
\bm{\Psi}_1\textbf{n}_1-\bm{\Psi}_3\textbf{n}_3=\textbf{\textsf{M}}^{-1}(\textbf{\textsf{r}}_3-\textbf{\textsf{r}}_1)\\
\:\;\;\;\;\:\;\;\;\;\:\;\;\;\;\:\;\;\;\;\:\;\;\vdots\\
\bm{\Psi}_1\textbf{n}_1-\bm{\Psi}_L\textbf{n}_L=\textbf{\textsf{M}}^{-1}(\textbf{\textsf{r}}_L-\textbf{\textsf{r}}_1)
\end{array}\right.
\end{equation}
and
\begin{equation}\label{ccc}
\left\{\begin{array}{ll}
\textbf{K}_{12}\textbf{n}_1-\textbf{K}_{21}\textbf{n}_2=\left(\textbf{\textsf{M}}\bm{\Psi}_{12}\right)^{-1}(\textbf{\textsf{r}}_2-\textbf{\textsf{r}}_1)\vspace{0.5ex} \\
\textbf{K}_{13}\textbf{n}_1-\textbf{K}_{31}\textbf{n}_3=\left(\textbf{\textsf{M}}\bm{\Psi}_{13}\right)^{-1}(\textbf{\textsf{r}}_3-\textbf{\textsf{r}}_1)\\
\:\;\;\;\;\:\;\;\;\;\:\;\;\;\;\:\;\;\;\;\:\;\;\;\;\vdots\\
\textbf{K}_{1L}\textbf{n}_1-\textbf{K}_{L1}\textbf{n}_L=\left(\textbf{\textsf{M}}\bm{\Psi}_{1L}\right)^{-1}(\textbf{\textsf{r}}_L-\textbf{\textsf{r}}_1),
\end{array}\right.
\end{equation}
in which $\textbf{K}_{1j}=\bm{\Psi}^{-1}_{1j}\bm{\Psi}_1$ and $\textbf{K}_{j1}=\bm{\Psi}^{-1}_{1j}\bm{\Psi}_j$ for $2\leq j\leq L$. From (\ref{cccd}) and (\ref{ccc}), $\left\{\textbf{\textsf{M}}^{-1}(\textbf{\textsf{r}}_i-\textbf{\textsf{r}}_1)\right\}_{i=2}^L$ and $\left\{\left(\textbf{\textsf{M}}\bm{\Psi}_{1i}\right)^{-1}(\textbf{\textsf{r}}_i-\textbf{\textsf{r}}_1)\right\}_{i=2}^L$ are all integer vectors; that is,
\begin{equation}
\textbf{\textsf{r}}_i-\textbf{\textsf{r}}_1\in\mathcal{L}(\textbf{\textsf{M}}\bm{\Psi}_{1i})\,\text{ for }2\leq i\leq L.
\end{equation}

For every $2\leq i\leq L$, we then estimate $\textbf{\textsf{r}}_i-\textbf{\textsf{r}}_1$ from the known erroneous remainders $\left\{\tilde{\textbf{\textsf{r}}}_i\right\}_{i=1}^{L}$ via finding a closest lattice point $\textbf{\textsf{v}}_i$ in $\mathcal{L}(\textbf{\textsf{M}}\bm{\Psi}_{1i})$ to $\tilde{\textbf{\textsf{r}}}_i-\tilde{\textbf{\textsf{r}}}_1$, i.e.,
\begin{equation}\label{xiaoxiao}
\textbf{\textsf{v}}_i=\argmin_{\textbf{\textsf{v}}\in\mathcal{L}(\textbf{\textsf{M}}\bm{\Psi}_{1i})}\,\lVert\textbf{\textsf{v}}-(\tilde{\textbf{\textsf{r}}}_i-\tilde{\textbf{\textsf{r}}}_1)\rVert,
\end{equation}
where $\tilde{\textbf{\textsf{r}}}_j\triangleq\textbf{\textsf{r}}_j+\triangle\textbf{\textsf{r}}_j\in\mathcal{F}(\textbf{\textsf{M}}\bm{\Psi}_j)$ for each $1\leq j\leq L$ is defined, and $\left\{\triangle\textbf{\textsf{r}}_i\right\}_{i=1}^{L}\in\mathbb{R}^{D}$ are the remainder errors.
We try to accurately determine $\left\{\bm{\Psi}_i\textbf{n}_i\right\}_{i=1}^L$. Specifically, we take the modulo-$\bm{\Psi}_i$ on both sides of the $(i-1)$-th equation in (\ref{cccd}) for $2\leq i\leq L$, and we have
\begin{equation}\label{yaogao}
\left\{\begin{array}{ll}
\bm{\Psi}_1\textbf{n}_1\equiv \textbf{0}  \!\!\mod \bm{\Psi}_1\\
\bm{\Psi}_1\textbf{n}_1\equiv \textbf{\textsf{M}}^{-1}(\textbf{\textsf{r}}_2-\textbf{\textsf{r}}_1)  \!\!\mod \bm{\Psi}_2\\
\bm{\Psi}_1\textbf{n}_1\equiv \textbf{\textsf{M}}^{-1}(\textbf{\textsf{r}}_3-\textbf{\textsf{r}}_1)  \!\!\mod \bm{\Psi}_3\\
\:\;\;\;\;\;\;\;\;\;\vdots\\
\bm{\Psi}_1\textbf{n}_1\equiv \textbf{\textsf{M}}^{-1}(\textbf{\textsf{r}}_L-\textbf{\textsf{r}}_1)  \!\!\mod \bm{\Psi}_L,\\
\end{array}\right.
\end{equation}
where the first equation spontaneously holds. Once $\left\{\textbf{\textsf{r}}_i-\textbf{\textsf{r}}_1\right\}_{i=2}^{L}$ are accurately estimated from (\ref{xiaoxiao}), i.e., $\textbf{\textsf{v}}_i=\textbf{\textsf{r}}_i-\textbf{\textsf{r}}_1$ for $2\leq i\leq L$, we can accurately determine $\bm{\Psi}_1\textbf{n}_1$ from (\ref{yaogao}) according to the MD-CRT (see Proposition \ref{pro1} above), provided that $\bm{\Psi}_1\textbf{n}_1\in$
$\mathcal{N}\left(\text{lcrm}(\bm{\Psi}_1,\bm{\Psi}_2,\cdots,\bm{\Psi}_L)\right)$, equivalently written as $\textbf{\textsf{M}}\bm{\Psi}_1\textbf{n}_1\in$
$\mathcal{F}\left(\textbf{\textsf{M}}\,\text{lcrm}(\bm{\Psi}_1,\bm{\Psi}_2,\cdots,\bm{\Psi}_L)\right)$, and also as
\begin{equation}
\lfloor\bm{\Psi}_1^{-1}\textbf{\textsf{M}}^{-1}\textbf{\textsf{m}}\rfloor\in\mathcal{N}\left(\bm{\Psi}_1^{-1}\text{lcrm}(\bm{\Psi}_1,\bm{\Psi}_2,\cdots,\bm{\Psi}_L)\right).
\end{equation}
Next, $\bm{\Psi}_i\textbf{n}_i$ can be accurately determined from (\ref{cccd}) as $\bm{\Psi}_1\textbf{n}_1-\textbf{\textsf{M}}^{-1}\textbf{\textsf{v}}_i$ for each $2\leq i\leq L$. One can see that the proposed robust MD-CRT for integer vectors/matrices (i.e., Theorem \ref{theo1}) and its
closed-form reconstruction algorithm (i.e., \textbf{Algorithm \ref{algo1}}) can be directly applied to (\ref{cccd}) (or (\ref{yaogao})). Thus, the following result is straightforwardly obtained.

\begin{corollary}
Let $\left\{\bm{\Psi}_i\right\}_{i=1}^{L}$ and $\textbf{\textsf{M}}$ in (\ref{bbb}) be $L$ different arbitrary nonsingular integer matrices and an arbitrary nonsingular real matrix, respectively.
Without loss of generality, we assume that the index $l_0$ with $1\leq l_0\leq L$ satisfies
\begin{equation}
\min_{1\leq j\leq L \atop j\neq l_0}\lambda_{\mathcal{L}(\textbf{\textsf{M}}\bm{\Psi}_{l_0j})}=\max_{1\leq i\leq L}\min_{1\leq j\leq L \atop j\neq i}\lambda_{\mathcal{L}(\textbf{\textsf{M}}\bm{\Psi}_{ij})}.
\end{equation}
For a real vector $\textbf{\textsf{m}}$ with
\begin{equation}
\lfloor\bm{\Psi}_{l_0}^{-1}\textbf{\textsf{M}}^{-1}\textbf{\textsf{m}}\rfloor\in\mathcal{N}\left(\bm{\Psi}_{l_0}^{-1}\text{lcrm}(\bm{\Psi}_1,\bm{\Psi}_2,\cdots,\bm{\Psi}_L)\right),
\end{equation}
we can accurately determine $\left\{\bm{\Psi}_i\textbf{n}_i\right\}_{i=1}^L$ from the erroneous rem-\\ainders $\left\{\tilde{\textbf{\textsf{r}}}_i\right\}_{i=1}^{L}$, if and only if
\begin{equation}
\textbf{0}=\argmin_{\textbf{\textsf{h}}\in\mathcal{L}(\textbf{\textsf{M}}\bm{\Psi}_{l_0j})}\lVert\textbf{\textsf{h}}-(\triangle\textbf{\textsf{r}}_j-\triangle\textbf{\textsf{r}}_{l_0})\rVert\; \text{ for }1\leq j\leq L \text{ and }j\neq l_0.
\end{equation}
Moreover, letting $\tau$ be the remainder error bound, i.e., $\lVert\triangle\textbf{\textsf{r}}_i\rVert\leq\tau$ for $1\leq i\leq L$, a simple sufficient condition is
\begin{equation}
\tau<\max\limits_{1\leq i\leq L}\min\limits_{1\leq j\leq L \atop j\neq i}\frac{\lambda_{\mathcal{L}(\textbf{\textsf{M}}\bm{\Psi}_{ij})}}{4}=\min\limits_{1\leq j\leq L \atop j\neq l_0}\frac{\lambda_{\mathcal{L}(\textbf{\textsf{M}}\bm{\Psi}_{l_0j})}}{4}.
\end{equation}
After $\left\{\bm{\Psi}_i\textbf{n}_i\right\}_{i=1}^L$ are accurately determined, a robust reconstruction $\tilde{\textbf{\textsf{m}}}$ of $\textbf{\textsf{m}}$ can be obtained by $\tilde{\textbf{\textsf{m}}}=\frac{1}{L}\sum_{i=1}^{L}\left(\textbf{\textsf{M}}\bm{\Psi}_i\textbf{n}_{i}+\tilde{\textbf{\textsf{r}}}_{i}\right)$ with $\lVert\tilde{\textbf{\textsf{m}}}-\textbf{\textsf{m}}\rVert\leq\tau$.
\end{corollary}

\section{Simulations}\label{sec6}
In this section, we first conduct some numerical simulations to verify the theoretical results of the robust MD-CRT in Theorem \ref{theo1} (see Sec. \ref{sec3} above), and then illustrate the performance of the robust MD-CRT in frequency estimation for a complex MD sinusoidal signal based on multiple sub-Nyquist samplers. For all experiments below, without loss of generality, we focus on the two-dimensional case, i.e., $D=2$, and the vector norm $\lVert\cdot\rVert$ involved is assumed to be the $\ell_2$ norm, i.e., $\lVert\cdot\rVert_2$.

\begin{figure}[H]
\centerline{\includegraphics[width=1\columnwidth,draft=false]{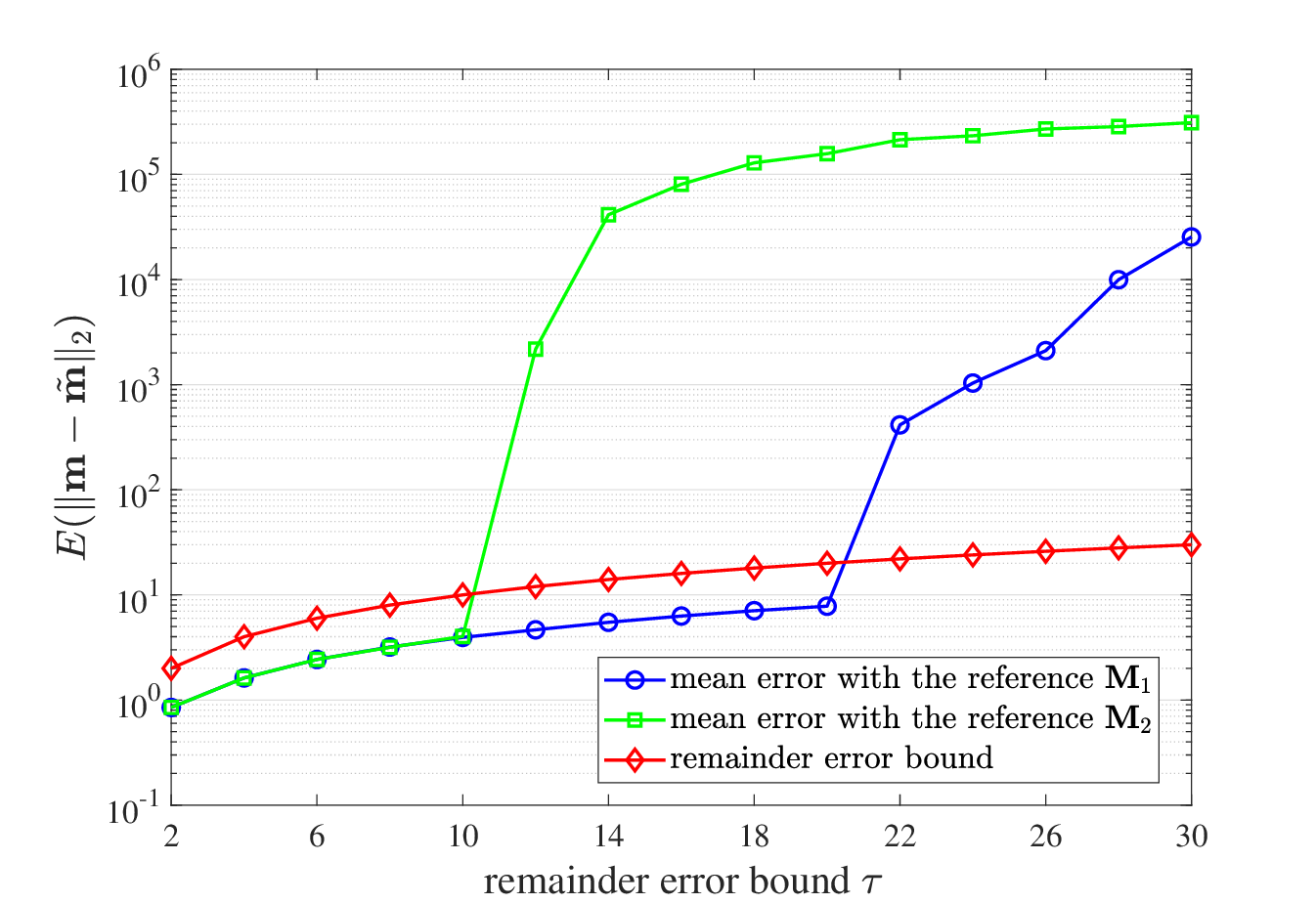}}
\caption{Mean error and theoretical error bound for the two cases using different reference moduli in \textbf{Algorithm \ref{algo1}}.}
 \label{fig1}
\end{figure}

We consider three moduli as $\textbf{M}_1=\left(
                  \begin{array}{cc}
                    1360 & 1788 \\
                    960 & 1728 \\
                  \end{array}
                \right)$, $\textbf{M}_2=\left(
                  \begin{array}{cc}
                    656 & 488 \\
                    256 & 448 \\
                  \end{array}
                \right)$, and $\textbf{M}_3=\left(
                  \begin{array}{cc}
                    1532 & 1576 \\
                    1392 & 1656 \\
                  \end{array}
                \right)$, which clearly do not satisfy the constraint (i.e., (\ref{specmod})) used in \cite{md-crt}.
We calculate an lcrm of $\left\{\textbf{M}_i\right\}_{i=1}^3$ as $\textbf{R}=\left(
                  \begin{array}{cc}
                    733248 & 540744 \\
                    655488 & 483264 \\
                  \end{array}
                \right)$, and the mini-\\mum distance of the lattice that is generated by a gcld of any pair of moduli as $\lambda_{\mathcal{L}(\textbf{M}_{12})}=85.0412$, $\lambda_{\mathcal{L}(\textbf{M}_{13})}=127.5617$, and $\lambda_{\mathcal{L}(\textbf{M}_{23})}=42.5206$. According to Theorem \ref{theo1}, we should choose $\textbf{M}_1$ as the reference moduli, i.e., $l_0=1$, and the reconstruction robustness bound is $85.0412/4=21.2603$. For comparison, we also choose $\textbf{M}_2$ as the reference moduli, and the reconstruction robustness bound is $42.5206/4=10.6302$. For these two cases, they have different reconstruction ranges.
                Let $\textbf{m}=\left(
                  \begin{array}{c}
                    515545 \\
                    460771  \\
                  \end{array}
                \right)$ be an integer vector we need to estimate, which obviously falls into the reconstruction ranges of the two cases. Therefore, with respect to each case, we investigate the remainder error bounds $\tau=0,2,4,\cdots,30$, and for each of them, we uniformly select the remainder errors $\lVert\triangle\textbf{r}_i\rVert_2\leq\tau, 1\leq i\leq 3$, and run $2000$ trails. For every trail, we utilize \textbf{Algorithm \ref{algo1}} to obtain one estimate $\tilde{\textbf{m}}$. In Fig. \ref{fig1}, we illustrate the mean error $E(\lVert\textbf{m}-\tilde{\textbf{m}}\rVert_2)$ in terms of different remainder error bounds for each of the two cases. One can see from Fig. \ref{fig1} that the performance is completely in line with the results of our proposed robust MD-CRT. That is to say, the mean error curve always lies beneath the remainder error bound curve when the remainder error bound is less than the reconstruction robustness bound, and then is about to break through the remainder error bound curve (i.e., robust reconst-\\ruction fails). Moreover, Fig. \ref{fig1} shows that choosing a proper modulus as the reference in \textbf{Algorithm \ref{algo1}} is beneficial to impr-\\oved robustness performance for the robust MD-CRT.

We next show a direct application of the robust MD-CRT to MD sinusoidal frequency estimation with multiple sub-Nyquist samplers in noise. Without loss of generality, assume that $\textbf{f}\in\mathbb{Z}^D$ is an unknown $D$-dimensional integer frequency of interest in a complex MD sinusoidal signal $x(\textbf{t})$ with noise $\omega(\textbf{t})$, i.e.,
\begin{equation}
x(\textbf{t})=e^{j2\pi\textbf{f}^{T}\textbf{t}}+\omega(\textbf{t}),\; \textbf{t}\in\mathbb{R}^D.
\end{equation}
It is known that samples of an MD signal are in general taken at vertex points of a sampling lattice (generated by a sampling matrix). We let $\left\{\textbf{M}_i^{-T}\right\}_{i=1}^L$ be $L$ different sampling matrices with corresponding sampling densities $\left\{|\text{det}(\textbf{M}_i)|\right\}_{i=1}^L$, where $\left\{\textbf{M}_i\right\}_{i=1}^L$ $\in\mathbb{Z}^{D\times D}$ are nonsingular integer matrices. We get the sampled sinusoidal signal of $x(\textbf{t})$ with the sampling matrix $\textbf{M}_i^{-T}$ as
\begin{equation}\label{ccccccc}
x_i[\textbf{n}]=e^{j2\pi \textbf{f}^T\textbf{M}_i^{-T}\textbf{n}}+\omega_i[\textbf{n}],\; \textbf{n}\in\mathbb{Z}^D.
\end{equation}
The MD discrete Fourier transform (DFT) with respect to $\textbf{M}_i^{T}$ is then implemented on $x_i[\textbf{n}], \textbf{n}\in\mathcal{N}(\textbf{M}_i^{T})$ \cite{ppvbook}, and we have
\begin{align}\label{xiouxi}
X_i[\textbf{k}]&=\sum_{\textbf{n}\in\mathcal{N}(\textbf{M}_i^T)}e^{j2\pi \textbf{f}^T\textbf{M}_i^{-T}\textbf{n}}e^{-j2\pi\textbf{k}^T\textbf{M}_i^{-T}\textbf{n}}+\Omega_i[\textbf{k}]\nonumber\\
&=\sum_{\textbf{n}\in\mathcal{N}(\textbf{M}_i^T)}e^{-j2\pi (\textbf{k}-\textbf{f})^T\textbf{M}_i^{-T}\textbf{n}}+\Omega_i[\textbf{k}]\nonumber\\
&=|\text{det}(\textbf{M}_i)|\,\delta[\textbf{k}-\textbf{r}_i]+\Omega_i[\textbf{k}]
\end{align}
for $\textbf{k}\in\mathcal{N}(\textbf{M}_i)$,
where $\textbf{r}_i$ is the remainder of $\textbf{f}$ modulo $\textbf{M}_i$, i.e., $\textbf{r}_i=\langle\textbf{f}\rangle_{\textbf{M}_i}$, $\Omega_i[\textbf{k}]$ is the MD DFT of $\omega_i[\textbf{n}]$ with respect to $\textbf{M}_i^{T}$, and the last equation holds due to the unitarity property of the MD DFT \cite{unitary}. Note that $\delta[\textbf{n}]$ stands for the MD discrete delta function, which equals $1$ if $\textbf{n}=\textbf{0}$ and $0$ otherwise. Hence, the remainder $\textbf{r}_i$ can be accurately detected as the peak in the MD DFT magnitude of $x_i[\textbf{n}]$ in (\ref{xiouxi}), when the signal-to-noise ratio (SNR, quantified as $\text{SNR}=-10\log_{10}(2\sigma^2)\;\text{dB}$ where $\omega_i[\textbf{n}]$ in (\ref{ccccccc}) is zero-mean complex white Gaussian noise with variance $2\sigma^2$) is not too low. Accordingly, $\textbf{f}$ can be accurately obtained from the detected remainders $\left\{\textbf{r}_i\right\}_{i=1}^L$ based on the MD-CRT in Proposition \ref{pro1}, if $\textbf{f}\in\mathcal{N}(\textbf{R})$, where $\textbf{R}$ is an lcrm of $\left\{\textbf{M}_i\right\}_{i=1}^L$. At this point, the Nyquist sampling density defined by $|\text{det}(\textbf{R})|$ is considerably greater than the sampling densities $\left\{|\text{det}(\textbf{M}_i)|\right\}_{i=1}^L$.
More interestingly, when the SNR is not too high, the detected remainders are likely to have errors, and thereby our proposed robust MD-CRT in Theorem \ref{theo1} offers an efficient approach for robustly estimating $\textbf{f}$ from the erroneous remainders.

\begin{figure}[H]
\centerline{\includegraphics[width=1\columnwidth,draft=false]{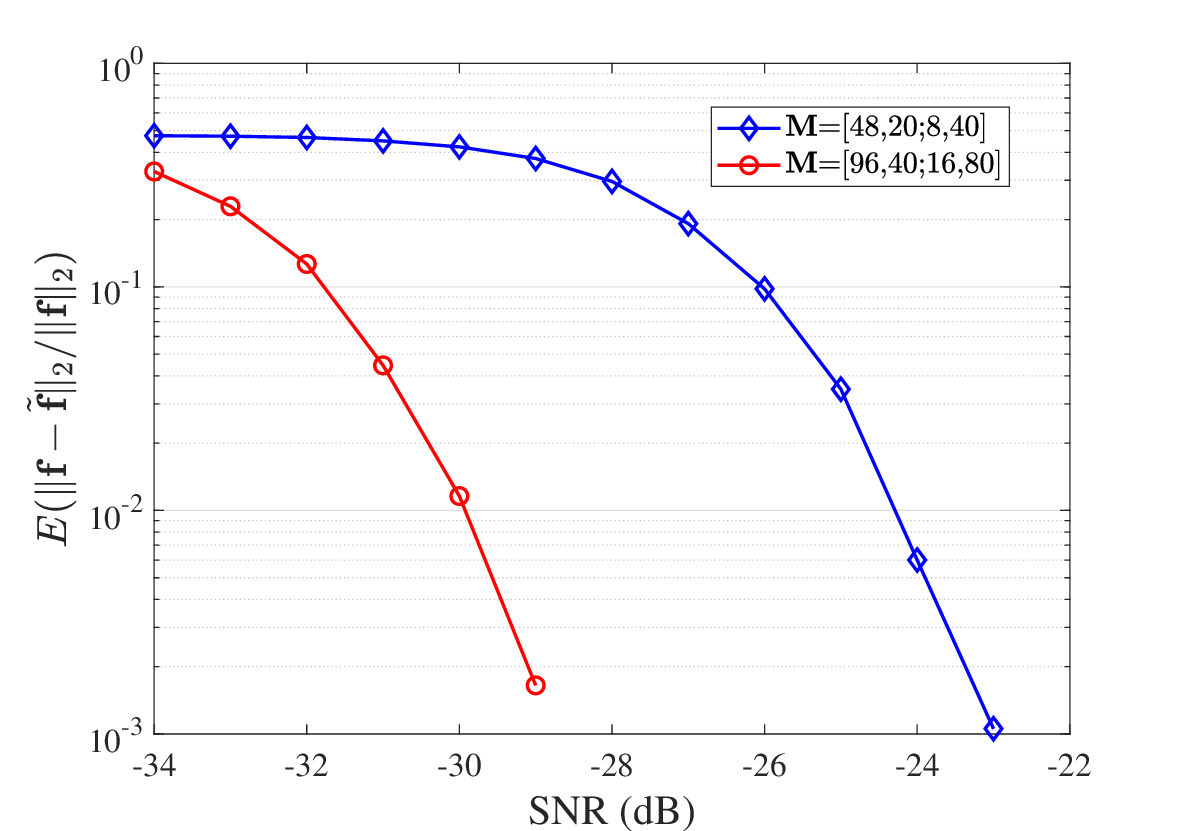}}
\caption{Mean relative error in terms of various SNR's for the two dif-\\ferent $\textbf{M}$'s.}
 \label{fig2}
\end{figure}

\begin{figure}[H]
\centerline{\includegraphics[width=1\columnwidth,draft=false]{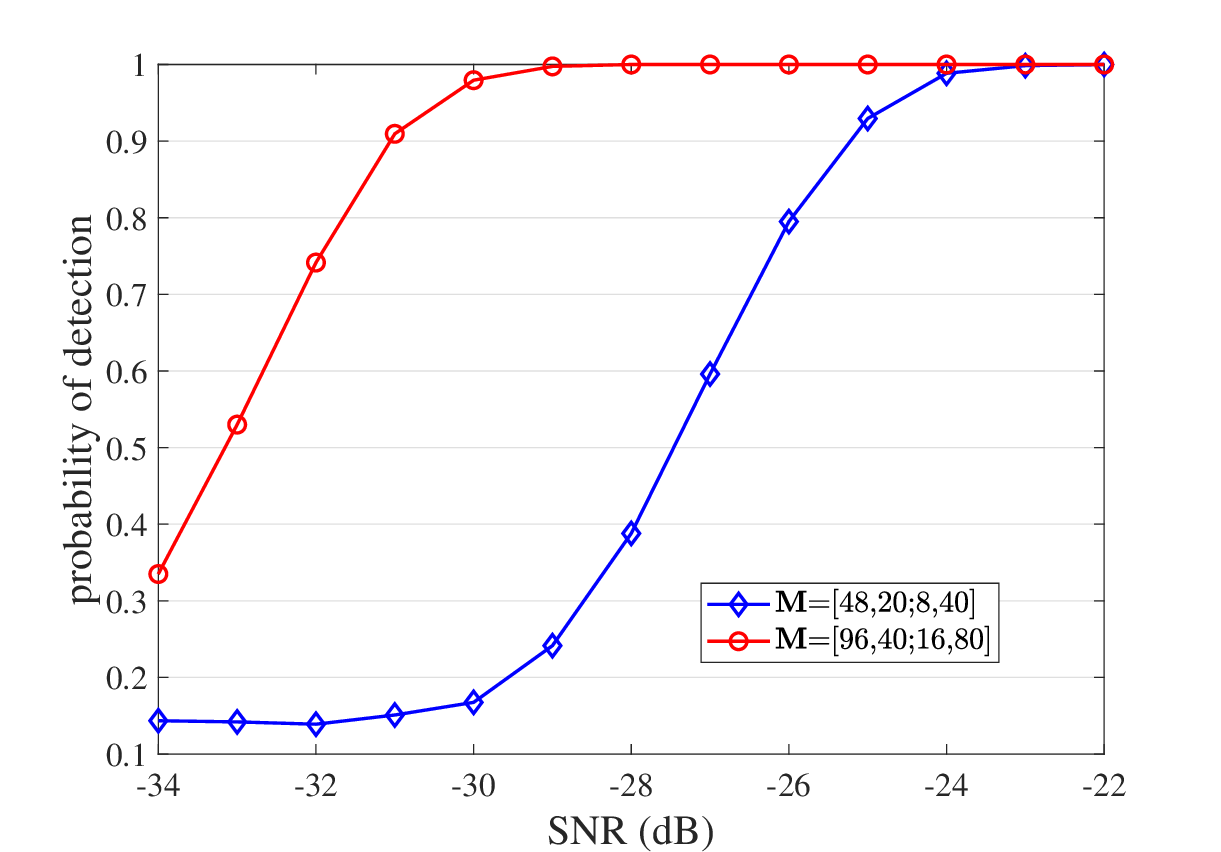}}
\caption{Probability of detection in terms of various SNR's for the two different $\textbf{M}$'s.}
 \label{fig3}
\end{figure}

To illustrate the performance of the robust MD-CRT in MD sinusoidal frequency estimation,  we consider two sampling m-\\atrices $\left\{\textbf{M}_i^{-T}\right\}_{i=1}^2$ for simplicity, where $\textbf{M}_i=\textbf{M}\bm{\Gamma}_i, i=1,2$, with
$\bm{\Gamma}_1=\left(
                  \begin{array}{cc}
                    2 & 1 \\
                    1 & 2 \\
                  \end{array}
                \right)$ and $\bm{\Gamma}_2=\left(
                  \begin{array}{cc}
                    2 & 2 \\
                    1 & 3 \\
                  \end{array}
                \right)$.
It is easily known that $\bm{\Gamma}_1$ and $\bm{\Gamma}_2$ are left coprime but not commutative. Hence, $\textbf{M}_1$ and $\textbf{M}_2$ do not satisfy the constraint (i.e., (\ref{specmod})) placed in \cite{md-crt}, and $\textbf{M}$ is their gcld. In these simulations, we investigate two cases of $\textbf{M}$, i.e., $\textbf{M}=\left(
                  \begin{array}{cc}
                    48 & 20 \\
                    8 & 40 \\
                  \end{array}
                \right)$ and $\textbf{M}=\left(
                  \begin{array}{cc}
                    96 & 40 \\
                    16 & 80 \\
                  \end{array}
                \right)$.
Based on the robust MD-CRT for the moduli $\textbf{M}_1$ and $\textbf{M}_2$, we reconstruct an MD frequency $\textbf{f}\in\mathbb{Z}^2$ from the detected remainders in the MD DFT domains of undersampled waveforms in (\ref{xiouxi}). From
Theorem \ref{theo1}, the two different $\textbf{M}$'s yield different reconstruction robustness bounds $10.6302$ and $21.2603$, respectively. We take $\textbf{f}=\left(
                  \begin{array}{c}
                    443  \\
                    388 \\
                  \end{array}
                \right)$, which is clearly within the reconstruction ranges of the two cases. In Fig. \ref{fig2}, we present the mean relative error $E(\lVert\textbf{f}-\tilde{\textbf{f}}\rVert_2/\lVert\textbf{f}\rVert_2)$ between $\textbf{f}$ and the reconstruction $\tilde{\textbf{f}}$ verse various SNR's for the two cases. Moreover, Fig. \ref{fig3} shows the probabil-\\ity of detection verse different SNR's to indicate the estimation accuracy for the two cases. In the experiments, we implement $2000$ trails for every SNR. From Fig. \ref{fig2} and Fig. \ref{fig3}, the second case with a larger reconstruction robustness bound results in better performance (i.e., lower mean relative error and higher probability of detection) than the first case with a smaller reco-\\nstruction robustness bound.

Furthermore, we compare two different sampling strategies with sampling matrices $\left\{\textbf{M}_i^{-T}\right\}_{i=1}^2$, where $\textbf{M}_1=\textbf{M}\bm{\Gamma}_1$ and $\textbf{M}_2=\textbf{M}\bm{\Gamma}_2$ are given as follows. \textit{Strategy \uppercase\expandafter{\romannumeral1}}: $\textbf{M}=\left(
                  \begin{array}{cc}
                    96 & 30 \\
                    12 & 90 \\
                  \end{array}
                \right)$, $\bm{\Gamma}_1=$ $\left(
                  \begin{array}{cc}
                    1 & 3 \\
                    3 & 1 \\
                  \end{array}
                \right)$, $\bm{\Gamma}_2=\left(
                  \begin{array}{cc}
                    5 & 2 \\
                    5 & 3 \\
                  \end{array}
                \right)$; and \textit{Strategy \uppercase\expandafter{\romannumeral2}}: $\textbf{M}=\left(
                  \begin{array}{cc}
                    10 & 32 \\
                    30 & 4 \\
                  \end{array}
                \right)$, $\bm{\Gamma}_1=\left(
                  \begin{array}{cc}
                    7 & 5 \\
                    5 & 7 \\
                  \end{array}
                \right)$, $\bm{\Gamma}_2=\left(
                  \begin{array}{cc}
                    5 & 1 \\
                    5 & 4 \\
                  \end{array}
                \right)$. It is easy to see that in each of the two strategies, the moduli do not satisfy the constraint (i.e., (\ref{specmod})) enforced in \cite{md-crt}, and $\textbf{M}=\text{gcld}(\textbf{M}_1,\textbf{M}_2)$. In addition, the two strategies possess the same Nyquist sampling density, i.e., share an identical lcrm $\textbf{R}=\left(
                  \begin{array}{cc}
                    -4782 & 5712 \\
                    -6894 & 8304 \\
                  \end{array}
                \right)$ with $|\text{det}(\textbf{R})|=331200$. According to Theorem \ref{theo1}, the two strategies have rec-\\
onstruction robustness bounds $23.7171$ and $7.9057$, respective-\\
ly. Let $\textbf{f}=\left(
                  \begin{array}{c}
                    810  \\
                    1181  \\
                  \end{array}
                \right)$, which simultaneously satisfies $\textbf{f}\in\mathcal{N}(\textbf{R})$ and falls into the reconstruction ranges of these two strategies. Fig. \ref{fig4} and Fig. \ref{fig5} illustrate the performance of the mean relative error and the probability of detection versus various SNR's for the two strategies, respectively, where $2000$ trails are implem-\\
ented for every SNR. As a consequence, \textit{Strategy \uppercase\expandafter{\romannumeral1}} achieves better performance than \textit{Strategy \uppercase\expandafter{\romannumeral2}}, while the sub-Nyquist sa-\\
mpling densities in \textit{Strategy \uppercase\expandafter{\romannumeral1}} are larger than those in \textit{Strategy \uppercase\expandafter{\romannumeral2}}, but far less than the Nyquist sampling density.

\begin{figure}[t]
\centerline{\includegraphics[width=1\columnwidth,draft=false]{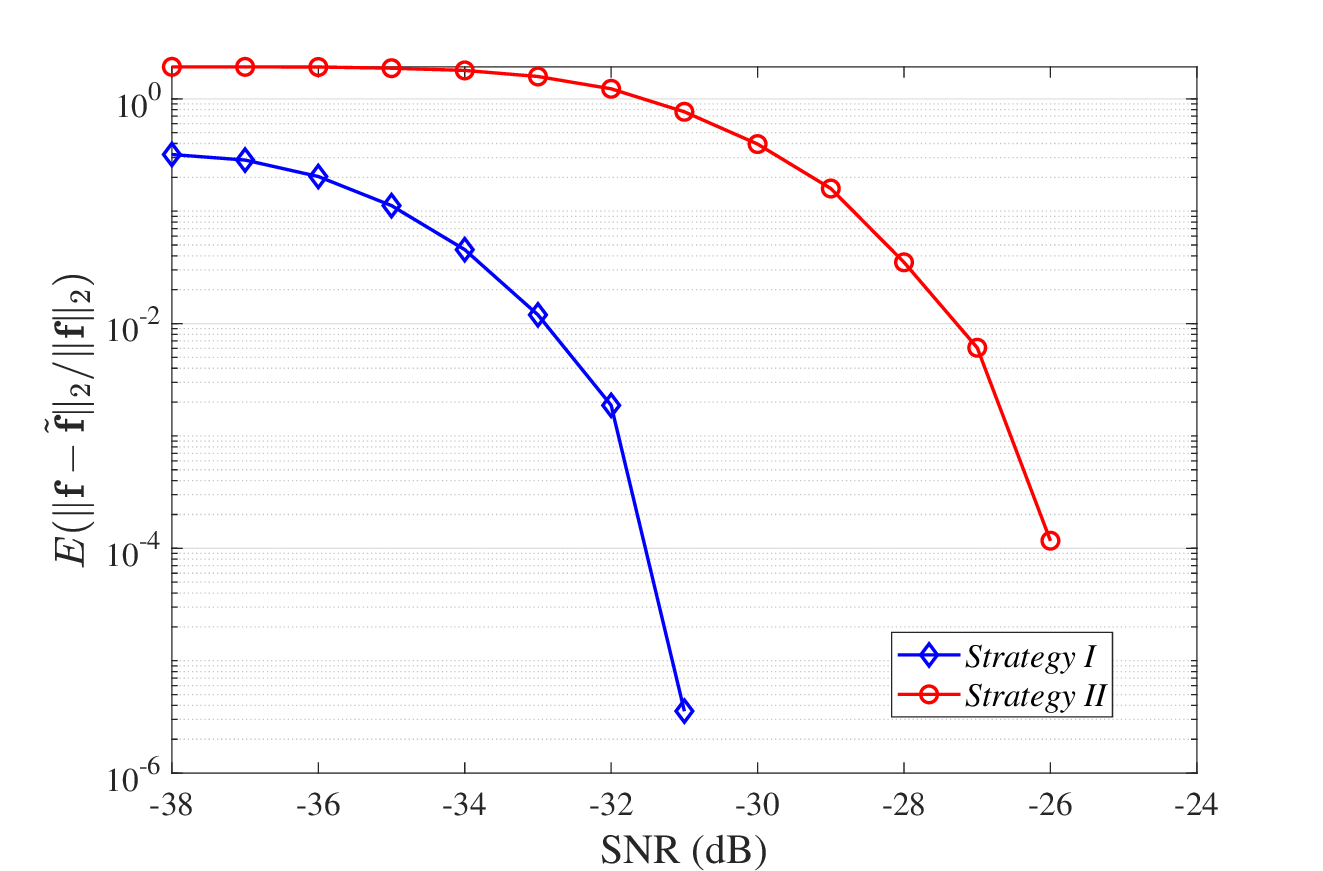}}
\caption{Mean relative error in terms of various SNR's for the two dif-\\ferent sampling strategies.}
 \label{fig4}
\end{figure}

\begin{figure}[t]
\centerline{\includegraphics[width=1\columnwidth,draft=false]{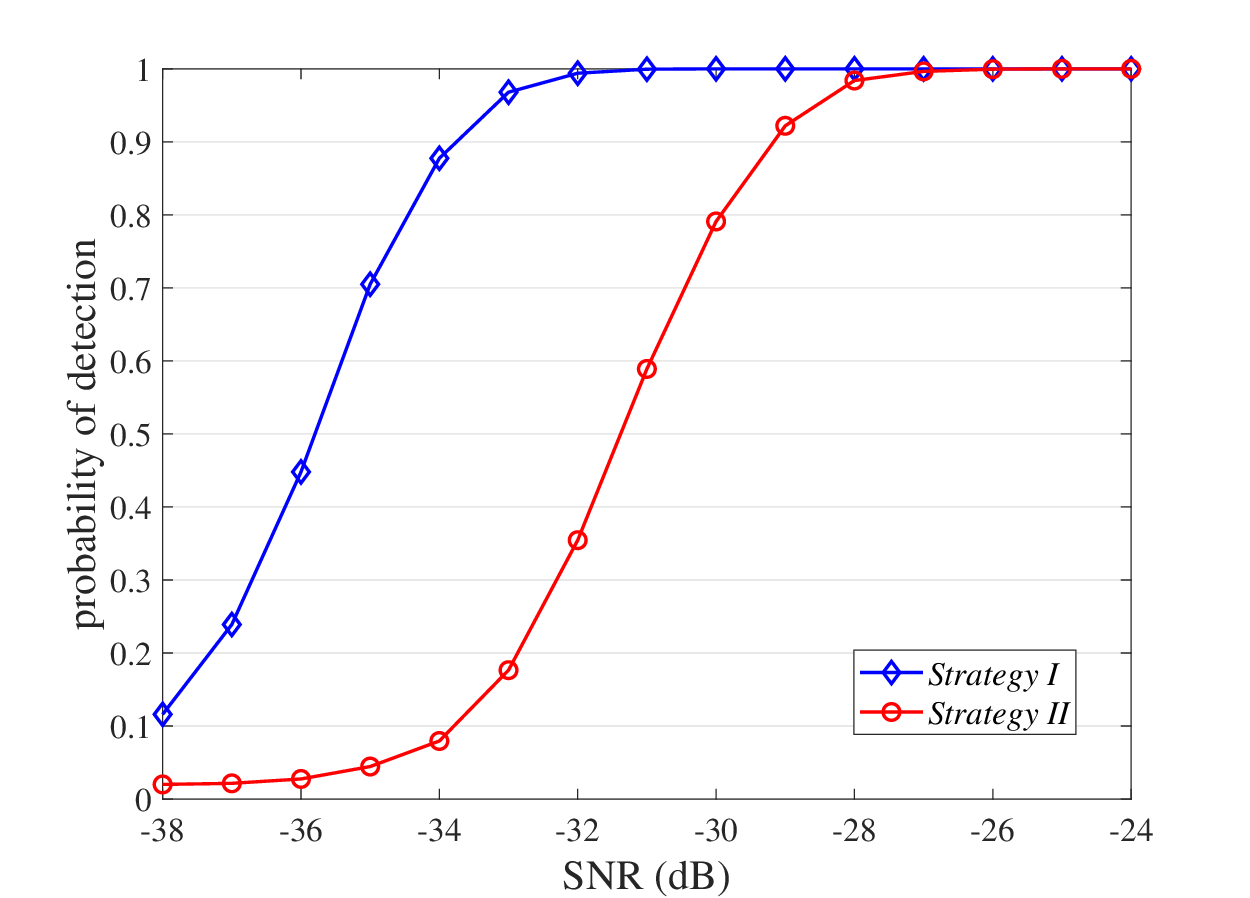}}
\caption{Probability of detection in terms of various SNR's for the two different sampling strategies.}
 \label{fig5}
\end{figure}

As a final comment, the release of the matrix commutativity and coprimeness constraint (used in \cite{md-crt}) on the moduli makes our proposed robust MD-CRT in this paper much more flexible for designing the optimal sampling matrices/lattices to achieve the best undersampling efficiency (e.g., the minimum sampling density as well as maximum robustness against noise). This is of great interest and will be studied in our future work.

\section{Conclusion}\label{sec7}
In this paper, we investigated the problem of robust reconstructions of an integer vector from the erroneous remainders. We introduced a theoretically well-founded solution to this pr-\\oblem by developing the robust MD-CRT for a general set of \\moduli that do not necessarily satisfy the strict constraint (i.e., the remaining integer matrices left-divided by a gcld of all the moduli are pairwise commutative and coprime) needed in the previous robust MD-CRT in \cite{md-crt}. Specifically, we first proved a necessary and sufficient condition on the difference between paired remainder errors, as well as a simple sufficient condition on the remainder error bound, for the robust MD-CRT for gen-\\eral
moduli, where a closed-form reconstruction algorithm was presented. We then generalized the proposed robust MD-CRT from integer vectors/matrices to real ones. We finally validated the robust MD-CRT for general moduli by conducting numerical simulations, and showed its performance in MD sinusoidal frequency estimation using multiple sub-Nyquist samplers. We believe that beyond MD sinusoidal frequency estimation from undersampled waveforms, the robust MD-CRT will have many other potential applications.

\end{document}